\documentclass[11pt]{article}
\usepackage{amsmath}   
\usepackage{amsthm}
\usepackage{amsfonts}
\usepackage{amssymb}    
\usepackage{xfrac} 
\usepackage{tabu}
\usepackage{bbm}

\usepackage{color}
\usepackage{tikz}
\usetikzlibrary{trees,snakes}
\usepackage{verbatim}
\usepackage{graphicx}
\usetikzlibrary{calc}
\usepackage{pgfplots}
\usepackage{cleveref}
\usepackage{graphicx}
\usepackage{subcaption}
\usepackage{mathrsfs} 
\usepackage{xcolor}

\newtheorem{theorem}{Theorem}
\newtheorem{lemma}{Lemma}
\newtheorem{proposition}{Proposition}
\newtheorem{definition}{Definition}
\newcommand{\nn}{\nonumber}

\newcommand{\dfn}{\stackrel{\triangle}{=}}

\newcommand {\reals} {\mathbb{R}}
\newcommand {\integers} {\mathbb{Z}}
\newcommand {\natur} {\mathbb{N}}
\newcommand {\Exp} {\mathbb{E}}
\newcommand {\prob} {\mathbb{P}}
\newcommand{\dint}{\mathrm{d}}

\newcommand {\bp} {\mbox{\boldmath $p$}}

\newcommand {\bs} {\mbox{\boldmath $s$}}

\newcommand {\bx} {\boldsymbol{x}}
\newcommand {\by} {\boldsymbol{y}}

\newcommand {\bP} {\mbox{\boldmath $P$}}

\newcommand {\bU} {\boldsymbol{U}}

\newcommand {\bX} {\boldsymbol{X}}
\newcommand {\bY} {\boldsymbol{Y}}

\newcommand {\balpha} {\boldsymbol{\alpha}}
\newcommand {\bbeta} {\boldsymbol{\beta}}
\newcommand {\bdelta} {\boldsymbol{\delta}}
\newcommand {\bmu} {\boldsymbol{\mu}}
\newcommand {\bGamma} {\boldsymbol{\Gamma}}
\newcommand {\bSigma} {\boldsymbol{\Sigma}}

\newcommand{\calB}{{\cal B}}

\newcommand{\calE}{{\cal E}}
\newcommand{\calF}{{\cal F}}

\newcommand{\calN}{{\cal N}}

\newcommand{\calS}{{\cal S}}

\newcommand{\calX}{{\cal X}}
\newcommand{\calY}{{\cal Y}}

\begin{document}
\thispagestyle{empty}
\title{Entropy Rate Bounds via Second-Order Statistics\\}
\author{\\ Ran Tamir \\}
\maketitle
	\begin{center}
		Signal and Information Processing Laboratory \\
		ETH Z\"urich, 8092 Z\"urich, Switzerland \\ 
		Email: tamir@isi.ee.ethz.ch
	\end{center}
	\vspace{1.5\baselineskip}
	\setlength{\baselineskip}{1.5\baselineskip}

\begin{abstract}
This work contains two single-letter upper bounds on the entropy rate of a discrete-valued stationary stochastic process, which only depend on second-order statistics, and are primarily suitable for models which consist of relatively large alphabets. The first bound stems from Gaussian maximum-entropy considerations and depends on the power spectral density (PSD) function of the process. While the PSD function cannot always be calculated in a closed-form, we also propose a second bound, which merely relies on some finite collection of auto-covariance values of the process. Both of the bounds consist of a one-dimensional integral, while the second bound also consists of a minimization problem over a bounded region, hence they can be efficiently calculated numerically. Examples are also provided to show that the new bounds outperform the standard conditional entropy bound. \\

\end{abstract}

\clearpage
\section{Introduction}

The Shannon entropy is one of the basic concepts of information theory, which provides a measure of uncertainty of a random variable. The information-theoretic definition of the entropy dates back to Shannon 1948's masterpiece \cite[Sec.\ 6]{Shannon}, where it has been proved to serve as a fundamental limit in lossless data compression. 
Ever since, the Shannon entropy has found its way to a numerous amount of different disciplines, ranging between machine learning \cite{IT_in_ML}, biology \cite{IT_in_Biology, IT_in_Biology2}, economics \cite{IT_in_Economics, IT_in_Economics2}, sociology \cite{IT_in_Sociology}, weather sciences \cite{IT_in_Metrology, IT_in_Metrology2}, and combinatorics \cite{IT_in_Combinatorics2, IT_in_Combinatorics}, for a non-exhaustive list. Although the definition of the Shannon entropy is fairly simple for a discrete random variable (as can be seen in \eqref{DEF_Shannon_Entropy} below), it cannot be calculated in closed-form for any such random variable. Even the Shannon entropies of some very common distributions, like the binomial distribution, the Poisson distribution, and the Borel distribution\footnote{A discrete random variable $X$ follows a Borel distribution if its PMF is given by
\begin{align}
P_{X}(n)= \frac{e^{-\mu n} (\mu n)^{n-1}}{n!},~~n=1,2,3,\ldots,~~\mu \in [0,1].
\end{align}
Although being somewhat less known, the Borel distribution is very common in queuing theory.}, do not admit closed-form expression. In such cases, one has merely two options: either to derive tight lower and upper bounds on the entropy \cite{Tight_Bounds1, Tight_Bounds2}, or to calculate it using a simple computer program. 

When it comes to stochastic processes, the problem becomes even more complicated, because a limiting operation is also involved. Because of that, only very few models have closed-form expressions for their entropy rates. For example, the derivation of the entropy rate of a finite order Markov process boils down to the calculation of a conditional entropy, which can also be quite exhausting in some models. Very special are the family of stationary Gaussian processes, whose differential entropy rate is given by a remarkable formula that only involves a one-dimensional integral over the logarithm of the power spectral density (PSD) function of the process \cite[p.\ 417]{Cover}. In some specific Gaussian models, like the auto-regressive moving-average (ARMA) parametric family of processes, the integral in this formula boils down to a very simple expression \cite{Ihara}. 
Similar results have been lately proved for other important Gaussian models that exhibit long-range dependence \cite{Feutrill}. 
Nonetheless, evaluation of exact values of entropy rates are a hard task in the general case. 
Still, in many cases of interest, like hidden Markov processes (HMPs) \cite{MERHAV2002}, the second order statistics can be relatively easily calculated (or estimated). 
Thus, the main objective of this work is to propose upper bounds on the entropy rates of discrete-valued stationary processes that relies merely on their second-order statistics.

The first result in this work is a generalization of a well-known upper bound on the entropy of a discrete random variable to discrete-valued stationary stochastic processes. This result stems from Gaussian maximum-entropy considerations and it only involves a one-dimensional integral over the logarithm of the PSD function of the process plus some constant. We demonstrate that this bound is applicable in a variety of models, like the discrete moving average (DMA) model \cite{Jacobs}, the quantized moving average (MA) model, and HMPs with relatively large or even infinite alphabets \cite{MacDonald}. Specifically, we show numerically that at least for the quantized MA model, the new bound is better than the standard conditional entropy bound.     

Since our first result depends on the PSD function of the process at hand, it cannot be useful in every case, since the calculation of the PSD function requires the entire (auto-)covariance function, and this alone is not always possible to derive in closed-form. Thus, the second result in this work is also a single-letter upper bound on the entropy rate, but this time, it depends merely on some finite set of the covariance function. 
In addition to a one-dimensional integral, this single-letter expression also involves a minimization problem in some bounded region, the dimension of which is identical to the number of covariance values in use. Nevertheless, both of these minimization and integration can be performed numerically quite efficiently. To exhibit the usefulness of our second result, we analyze the quantized-hidden auto-regressive (AR) process, and show that the new proposed upper bound is better than the standard conditional entropy bound. For this specific AR model, our first result is practically irrelevant, since it relies on a one-dimensional integral over the PSD function of the process, and in this case, the covariance function itself is already given by a relatively cumbersome expression.            

In the realm of discrete-valued stationary stochastic processes, not many specific bounds are known, except for the standard conditional entropy upper bound. 
One process that has been extensively studied over the past two decades is the binary HMP, which is formed by passing a simple binary Markov chain through a binary channel. Both lower and upper bounds have been derived in \cite{OW2011}, using a new approach for bounding the entropy rate of HMP by constructing an alternative Markov process corresponding to the log-likelihood ratio of estimating the positivity of the current hidden symbol $X_{0}$ based on the past observations $Y_{-\infty}^{0}$. The techniques of \cite{OW2011} have been used in \cite{NOW2005} to obtain tight bounds for the entropy rate in the rare-transition regime.
Entropy rates for HMPs with rare transitions have also been studied in \cite{PQ2011}.    
More refined lower bounds on the entropy rate of binary HMPs have been derived in \cite{OS2015} using a minimum-mean square error approach, and in \cite{O2016}, relying on a strengthened version of Mrs.\ Gerber's Lemma. 
Entropy rates of HMPs formed by passing a binary Markov chain through an arbitrary memoryless channel was studied in \cite{LuoGuo2009}. 
Since the entropy rate of a HMP is closely related to the maximal Lyapunov exponent, a connection that was observed and discussed in \cite{JSS2008}, bounds on the entropy rate of a HMP can be deduced immediately from existing bounds on the maximal Lyapunov exponent, e.g., \cite{Lyapunov2021}.      
A review of entropy rate estimation can be found in \cite{Estimation}.

The remaining part of the paper is organized as follows. 
In Section \ref{SEC2}, we establish notation conventions.
In Section \ref{SEC3}, we review some preliminaries, provide motivation for this research, and formalize the main objectives of this work. 
In Section \ref{SEC4}, we provide and discuss the main results, and in Section \ref{SEC_Applications}, we exhibit the usefulness of our results via two specific quantized processes.
In the Appendixes, we prove our results.

\section{Notation Conventions} \label{SEC2}

Throughout the paper, random variables will be denoted by capital letters, realizations will be denoted by the corresponding lower case letters, and their alphabets will be denoted by calligraphic letters. 
Random vectors and their realizations will be denoted, 
respectively, by boldface capital and lower case letters. 
Their alphabets will be superscripted by their dimensions. 
%Sources and channels will be subscripted by the names of the relevant random 
%variables/vectors and their conditionings, whenever applicable, 
%following the standard notation conventions, e.g., $Q_{X}$, $Q_{Y|X}$, and so on. 
%When there is no room for ambiguity, these subscripts will be omitted. 
%For a generic joint distribution $Q_{XY} = \{Q_{XY}(x,y), x \in \mathcal{X}, y \in \mathcal{Y} \}$, which will often 
%be abbreviated by $Q$, 
%$H_{Q}(Y|X)$ will stand for the conditional entropy of $Y$ given $X$, i.e.,
%\begin{align}
%H_{Q}(Y|X) = - \sum_{(x,y) \in \calX \times \calY} Q_{XY}(x,y) \log Q_{Y|X}(y|x).
%\end{align}
%The Kullback--Leibler divergence function between two probability distributions $P$ and $Q$ is defined as
%\begin{align} \label{DEF_Bin_DIVERGENCE}
%D(P \| Q) = \sum_{x \in \calX} P(x) \log \frac{P(x)}{Q(x)},
%\end{align}
%where logarithms, here and throughout the sequel, are understood to be taken to the natural base.
%The weighted divergence between 
%two conditional distributions (channels), say, $Q_{Y|X}$ and $W = \{W(y|x), x \in \calX, y \in \calY \}$, with weighting $Q_{X}$ is defined as
%\begin{align} \label{DEF_DIVERGENCE}
%D(Q_{Y|X} || W | Q_{X}) 
%= \sum_{x \in \calX} Q_{X}(x) 
%\sum_{y \in \calY} Q_{Y|X}(y|x) \log \frac{Q_{Y|X}(y|x)}{W(y|x)}.
%\end{align}
The cumulative distribution function of a standard normal random variable is defined by
\begin{align} \label{DEF_PHI}
\Phi(t) = \int_{-\infty}^{t} \frac{1}{\sqrt{2\pi}} \exp\left\{-\frac{s^{2}}{2}\right\} ds.
\end{align}
The probability of an event $\calE$ will be denoted by $\prob \{\calE\}$, and the expectation operator with respect to a 
probability distribution $Q$ will be denoted by $\mathbb{E}_{Q} [\cdot]$, where the subscript will often be omitted.
The variance of a random variable $X$ is denoted by $\text{Var}[X]$. 
%The indicator function of an event $\calA$ will be denoted by $\IND\{\calA\}$. 
%The set $\{1,2,\ldots,n\}$ will often be denoted by $[1:n]$. 

For a wide-sense stationary process $\{X_{n}\}_{n \geq 0}$ with mean $\mu_{\mbox{\tiny X}} = \Exp[X_{n}]$, the covariance function will be denoted by 
\begin{align}
R_{\mbox{\tiny X}}(k) = \Exp[(X_{n+k}-\mu_{\mbox{\tiny X}})(X_{n}-\mu_{\mbox{\tiny X}})],
\end{align}
and the PSD function is defined by
\begin{align} \label{Def_PSD}
\Phi_{\mbox{\tiny X}}(\lambda) 
= \sum_{k=-\infty}^{\infty} R_{\mbox{\tiny X}}(k) e^{i \lambda k},
\end{align} 
when $\{R_{\mbox{\tiny X}}(k)\}$ is absolutely summable.

\section{Preliminaries, Motivation, and Objectives} \label{SEC3}
\subsection{Preliminaries} 
In this paper, we will be using the natural logarithm in all of the definitions, and hence the units of entropy that we will be working with are nats. We include standard definitions in order that all notations be precisely defined.

\begin{definition}
	The Shannon entropy, $H(X)$, of a discrete random variable $X$ with probability mass function $P_{X}$, is defined as
	\begin{align} \label{DEF_Shannon_Entropy}
	H(X) = - \sum_{x \in \calX} P_{X}(x) \log P_{X}(x),
	\end{align}
	where $\calX$ is the (possibly infinite) alphabet of the random variable.
\end{definition}

The Shannon entropy can be extended into the multivariate case and hence to stochastic processes using the joint entropy for a collection of random variables.

\begin{definition}
	The joint entropy of a random vector $\bX = (X_{1},X_{2},\ldots,X_{n})$ with probability mass function $P_{\bX}$, is defined as 
\begin{align}
H(X_{1},X_{2},\ldots,X_{n}) = -\sum_{\bx \in \Omega} P_{\bX}(\bx) \log P_{\bX}(\bx),
\end{align}
where $\Omega = \calX_{1} \times \calX_{2} \times \ldots \times \calX_{n}$ is the support of the random vector, which is usually $\calX^{n}$.
\end{definition}

In this work, we also need the continuous extension of Shannon's joint entropy.  

\begin{definition}
	The joint differential entropy of a random vector $\bX = (X_{1},X_{2},\ldots,X_{n})$ with probability density function $f_{\bX}$, is defined as 
	\begin{align}
	h(X_{1},X_{2},\ldots,X_{n}) = -\int_{\Omega} f_{\bX}(\bx) \log f_{\bX}(\bx) \dint \bx,
	\end{align}
	where $\Omega \subseteq \reals^{n}$ is the support of the random vector.
\end{definition}

Finally, we define the concept of entropy rate, which can be thought of as the average amount of new information from each sample of a random variable in a discrete stochastic process.

\begin{definition}
	The entropy rate of a discrete stochastic process $\{X_{n}\}_{n \geq 1}$ is defined by
	\begin{align}
	\bar{H}(X) = \lim_{n \to \infty} \frac{H(X_{1},X_{2},\ldots,X_{n})}{n},
	\end{align}
	provided that the limit exists.
\end{definition} 
For stationary stochastic processes, the entropy rate is well defined \cite[Theorem 4.2.1]{Cover}.

Although this work mainly concerns discrete-valued stochastic processes, the two main applications that will be discussed in details in Section \ref{SEC_Applications} have a pure Gaussian mechanism, hence we will be needing the following definition.

\begin{definition}
	A stochastic process is called a Gaussian process if and only if every finite collection of random variables from the stochastic process has a multivariate Gaussian distribution. That is, for every $k \geq 1$ and every $(t_{1}, \ldots, t_{k}) \in \integers^{k}$,
	\begin{align}
	(X_{t_{1}}, \ldots, X_{t_{k}}) \sim \calN(\bmu,\bSigma),
	\end{align}
	where $\bmu$ is the vector of expected values and $\bSigma$ is the covariance matrix.
\end{definition}  

For a stationary Gaussian process with a PSD function $\Phi_{\mbox{\tiny X}}(\lambda)$, Kolmogorov showed that the differential entropy rate can be expressed as \cite[p.\ 417]{Cover}    
\begin{align} \label{Gaussian_Entropy_Rate}
\bar{H}(X) 
&= \frac{1}{2} \log (2 \pi e) 
+ \frac{1}{4\pi} \int_{0}^{2\pi} \log \Phi_{\mbox{\tiny X}}(\lambda) \dint\lambda.
\end{align}

The results in this work are based on a simple relation between the entropy of a discrete-valued random vector and the differential entropy of some related continuous-valued random vector.
Specifically, let $\bY_{n} = (Y_{1}, Y_{2}, \ldots, Y_{n})$ be a discrete-valued random vector. 
Let $\bU_{n} = (U_{1}, U_{2}, \ldots, U_{n})$ be a random vector with independent entries, such that $U_{i}$ is uniformly distributed in $[0,1)$ for any $1 \leq i \leq n$. Define the sum $\tilde{\bY}_{n} = \bY_{n} + \bU_{n}$. Then, the following result is going to be instrumental in the proofs of our main results.
\begin{lemma} \label{Basic_Lemma}
	It holds that $H(\bY_{n})=h(\tilde{\bY}_{n})$.
\end{lemma}

\begin{proof}
	For $\by = (y_{1},y_{2},\ldots,y_{n}) \in \calY^{n}$, denote the hypercube
	\begin{align}
	\calB(\by) = [y_{1},y_{1}+1) \times [y_{2},y_{2}+1) \times \ldots \times [y_{n},y_{n}+1).
	\end{align}
	Now,
	\begin{align}
	H(\bY_{n}) 
	&= -\sum_{\by \in \calY^{n}} p(\by) \log p(\by) \\
	&= -\sum_{\by \in \calY^{n}} \int_{\calB(\by)} f_{\tilde{\bY}}(\bs) \log f_{\tilde{\bY}}(\bs) d\bs \\
	&= - \int_{\reals^{n}} f_{\tilde{\bY}}(\bs) \log f_{\tilde{\bY}}(\bs) d\bs \\
	\label{ref3}
	&= h(\tilde{\bY}_{n}),
	\end{align}
	since $f_{\tilde{\bY}}(\bs) = p(\by)$ for any $\bs \in \calB(\by)$.	
\end{proof}

\subsection{Motivation and Objectives}

Consider the following differential entropy bound on discrete entropy \cite[Problem 8.7]{Cover}
\begin{align} \label{Univariate_Bound}
H(X) \leq \frac{1}{2} \log \left[2 \pi e \left(\text{Var}(X) + \frac{1}{12}\right)\right],
\end{align}
which was attributed to the independent (unpublished) works of Massey and Willems. However, in 1975 Djackov \cite{Djackov} had already published the bound in connection with his work on coin-weighing. 
For some random variables, mainly with relatively large or infinite alphabets, the Shannon entropy does not admit a closed form expression, but still, their variance can relatively easy be calculated. In such cases, the bound in \eqref{Univariate_Bound} may be a good compromise between tightness and ease of calculation. As an example, consider a random variable $X$ with a Poisson distribution. The probability mass function of such a random variable is given by    
\begin{align} \label{Poisson_PMF}
P_{X}(k) = \frac{\lambda^{k}e^{-\lambda}}{k!},~~k=0,1,2,\ldots,
\end{align}
where $\lambda \in (0,\infty)$ is a given parameter. In this case, a direct substitution of \eqref{Poisson_PMF} into \eqref{DEF_Shannon_Entropy} yields
\begin{align}
H_{\mbox{\tiny Poisson}}(X) = \lambda[1-\log(\lambda)] + e^{-\lambda}\sum_{k=0}^{\infty} \frac{\lambda^{k}\log(k!)}{k!},
\end{align}  
which although lands itself to numerical evaluation, this expression is not so easy to study as a function of $\lambda$. However, the variance of the Poisson random variable is simply $\lambda$ and thus   
\begin{align} \label{Poisson_UpperBound}
H_{\mbox{\tiny Poisson}}(X) \leq \frac{1}{2} \log \left[2 \pi e \left(\lambda + \frac{1}{12}\right)\right].
\end{align}
As can be seen in Figure \ref{fig:PoissonGap}, for relatively large $\lambda$ values, the gap between the exact entropy and its upper bound is quite small.
\begin{figure}[ht!]
	\centering
	\begin{tikzpicture}[scale=1.3]
	\begin{axis}[
	disabledatascaling,
	%x=25cm,
	%y=30cm,
	scaled x ticks=false,
	xticklabel style={/pgf/number format/fixed,
		/pgf/number format/precision=3},
	scaled y ticks=false,
	yticklabel style={/pgf/number format/fixed,
		/pgf/number format/precision=3},
	xlabel={$\lambda$},
	xmin=0, xmax=10,
	ymin=0, ymax=2.6,
	legend pos=south east,
	%    ymajorgrids=true,
	%    grid style=dashed,
	]
	
	\addplot[smooth,color=black!30!green,thick]
	table[row sep=crcr] 
	{
0	0	\\
0.1	0.333676997	\\
0.2	0.535377973	\\
0.3	0.691143959	\\
0.4	0.819071204	\\
0.5	0.927637467	\\
0.6	1.021752461	\\
0.7	1.104597199	\\
0.8	1.178383216	\\
0.9	1.244724633	\\
1	1.304842242	\\
1.1	1.359684844	\\
1.2	1.41000589	\\
1.3	1.456414238	\\
1.4	1.499409066	\\
1.5	1.539404653	\\
1.6	1.576748421	\\
1.7	1.611734377	\\
1.8	1.644613311	\\
1.9	1.675600656	\\
2	1.704882644	\\
2.1	1.73262118	\\
2.2	1.758957749	\\
2.3	1.784016575	\\
2.4	1.807907199	\\
2.5	1.830726608	\\
2.6	1.852560991	\\
2.7	1.87348721	\\
2.8	1.893574033	\\
2.9	1.912883179	\\
3	1.931470198	\\
3.1	1.949385227	\\
3.2	1.966673638	\\
3.3	1.983376591	\\
3.4	1.999531514	\\
3.5	2.015172523	\\
3.6	2.030330777	\\
3.7	2.045034799	\\
3.8	2.05931075	\\
3.9	2.073182672	\\
4	2.0866727	\\
4.1	2.099801249	\\
4.2	2.112587185	\\
4.3	2.125047967	\\
4.4	2.137199781	\\
4.5	2.149057657	\\
4.6	2.160635572	\\
4.7	2.171946542	\\
4.8	2.183002707	\\
4.9	2.193815406	\\
5	2.204395243	\\
5.1	2.214752149	\\
5.2	2.224895434	\\
5.3	2.234833841	\\
5.4	2.244575584	\\
5.5	2.254128397	\\
5.6	2.263499563	\\
5.7	2.272695953	\\
5.8	2.281724051	\\
5.9	2.290589988	\\
6	2.299299563	\\
6.1	2.307858268	\\
6.2	2.316271308	\\
6.3	2.324543622	\\
6.4	2.332679899	\\
6.5	2.340684596	\\
6.6	2.348561953	\\
6.7	2.356316007	\\
6.8	2.363950602	\\
6.9	2.371469405	\\
7	2.378875915	\\
7.1	2.386173473	\\
7.2	2.39336527	\\
7.3	2.400454361	\\
7.4	2.407443665	\\
7.5	2.41433598	\\
7.6	2.421133986	\\
7.7	2.427840253	\\
7.8	2.434457243	\\
7.9	2.440987324	\\
8	2.447432767	\\
8.1	2.453795754	\\
8.2	2.460078386	\\
8.3	2.466282681	\\
8.4	2.472410583	\\
8.5	2.478463962	\\
8.6	2.484444621	\\
8.7	2.490354298	\\
8.8	2.496194668	\\
8.9	2.501967347	\\
9	2.507673896	\\
9.1	2.513315821	\\
9.2	2.518894576	\\
9.3	2.524411568	\\
9.4	2.529868157	\\
9.5	2.535265656	\\
9.6	2.540605338	\\
9.7	2.545888434	\\
9.8	2.551116135	\\
9.9	2.556289596	\\
10	2.561409935	\\
10.1	2.566478237	\\
10.2	2.571495552	\\
10.3	2.5764629	\\
10.4	2.581381269	\\
10.5	2.58625162	\\
10.6	2.591074883	\\
10.7	2.595851963	\\
10.8	2.60058374	\\
10.9	2.605271066	\\
11	2.609914773	\\
11.1	2.614515665	\\
11.2	2.619074527	\\
11.3	2.623592124	\\
11.4	2.628069195	\\
11.5	2.632506465	\\
11.6	2.636904636	\\
11.7	2.641264393	\\
11.8	2.645586402	\\
11.9	2.649871313	\\
12	2.654119759	\\
12.1	2.658332356	\\
12.2	2.662509706	\\
12.3	2.666652394	\\
12.4	2.670760993	\\
12.5	2.674836059	\\
12.6	2.678878137	\\
12.7	2.682887756	\\
12.8	2.686865437	\\
12.9	2.690811683	\\
13	2.694726988	\\
13.1	2.698611836	\\
13.2	2.702466696	\\
13.3	2.70629203	\\
13.4	2.710088285	\\
13.5	2.713855903	\\
13.6	2.717595312	\\
13.7	2.721306931	\\
13.8	2.724991173	\\
13.9	2.728648437	\\
14	2.732279118	\\
14.1	2.735883598	\\
14.2	2.739462254	\\
14.3	2.743015454	\\
14.4	2.746543557	\\
14.5	2.750046917	\\
14.6	2.753525878	\\
14.7	2.756980777	\\
14.8	2.760411947	\\
14.9	2.763819711	\\
15	2.767204387	\\
15.1	2.770566285	\\
15.2	2.773905711	\\
15.3	2.777222962	\\
15.4	2.780518333	\\
15.5	2.78379211	\\
15.6	2.787044574	\\
15.7	2.790276002	\\
15.8	2.793486664	\\
15.9	2.796676825	\\
16	2.799846746	\\
16.1	2.802996683	\\
16.2	2.806126885	\\
16.3	2.809237599	\\
16.4	2.812329066	\\
16.5	2.815401524	\\
16.6	2.818455204	\\
16.7	2.821490335	\\
16.8	2.824507141	\\
16.9	2.827505842	\\
17	2.830486655	\\
17.1	2.833449791	\\
17.2	2.836395459	\\
17.3	2.839323865	\\
17.4	2.842235209	\\
17.5	2.845129689	\\
17.6	2.8480075	\\
17.7	2.850868833	\\
17.8	2.853713875	\\
17.9	2.85654281	\\
18	2.859355822	\\
18.1	2.862153086	\\
18.2	2.864934781	\\
18.3	2.867701077	\\
18.4	2.870452144	\\
18.5	2.873188149	\\
18.6	2.875909257	\\
18.7	2.878615628	\\
18.8	2.881307422	\\
18.9	2.883984795	\\
19	2.886647901	\\
19.1	2.889296891	\\
19.2	2.891931913	\\
19.3	2.894553115	\\
19.4	2.897160641	\\
19.5	2.899754633	\\
19.6	2.902335231	\\
19.7	2.904902572	\\
19.8	2.907456791	\\
19.9	2.909998024	\\
20	2.9125264	\\
	};
	\legend{}
	\addlegendentry{$H_{\mbox{\tiny Poisson}}(X)$}	
	
	\addplot[smooth,color=black!20!purple,thick,dash pattern={on 3pt off 2pt}]
	table[row sep=crcr]
	{
0	0.176485208	\\
0.1	0.570713888	\\
0.2	0.788372924	\\
0.3	0.93951336	\\
0.4	1.055414167	\\
0.5	1.149440283	\\
0.6	1.228552285	\\
0.7	1.296840053	\\
0.8	1.356912209	\\
0.9	1.410534974	\\
1	1.458959887	\\
1.1	1.503106191	\\
1.2	1.543668963	\\
1.3	1.581186556	\\
1.4	1.616084437	\\
1.5	1.648704698	\\
1.6	1.679326511	\\
1.7	1.708180669	\\
1.8	1.735460161	\\
1.9	1.761327999	\\
2	1.785923121	\\
2.1	1.809364914	\\
2.2	1.831756715	\\
2.3	1.853188567	\\
2.4	1.873739405	\\
2.5	1.893478811	\\
2.6	1.912468435	\\
2.7	1.930763158	\\
2.8	1.948412049	\\
2.9	1.965459155	\\
3	1.981944165	\\
3.1	1.997902966	\\
3.2	2.013368116	\\
3.3	2.028369242	\\
3.4	2.042933378	\\
3.5	2.057085266	\\
3.6	2.070847603	\\
3.7	2.084241261	\\
3.8	2.097285479	\\
3.9	2.109998028	\\
4	2.122395357	\\
4.1	2.134492722	\\
4.2	2.146304295	\\
4.3	2.157843268	\\
4.4	2.169121942	\\
4.5	2.180151801	\\
4.6	2.190943587	\\
4.7	2.20150736	\\
4.8	2.211852557	\\
4.9	2.221988039	\\
5	2.23192214	\\
5.1	2.241662708	\\
5.2	2.251217139	\\
5.3	2.260592414	\\
5.4	2.269795127	\\
5.5	2.278831518	\\
5.6	2.287707491	\\
5.7	2.296428642	\\
5.8	2.305000281	\\
5.9	2.313427446	\\
6	2.321714929	\\
6.1	2.329867283	\\
6.2	2.337888846	\\
6.3	2.345783747	\\
6.4	2.353555924	\\
6.5	2.361209135	\\
6.6	2.368746966	\\
6.7	2.376172845	\\
6.8	2.383490049	\\
6.9	2.390701712	\\
7	2.397810837	\\
7.1	2.404820297	\\
7.2	2.41173285	\\
7.3	2.418551137	\\
7.4	2.425277696	\\
7.5	2.431914962	\\
7.6	2.438465274	\\
7.7	2.444930881	\\
7.8	2.451313946	\\
7.9	2.457616551	\\
8	2.463840698	\\
8.1	2.469988316	\\
8.2	2.476061265	\\
8.3	2.482061337	\\
8.4	2.48799026	\\
8.5	2.493849702	\\
8.6	2.499641273	\\
8.7	2.505366526	\\
8.8	2.511026964	\\
8.9	2.516624038	\\
9	2.522159149	\\
9.1	2.527633657	\\
9.2	2.533048872	\\
9.3	2.538406066	\\
9.4	2.543706469	\\
9.5	2.548951272	\\
9.6	2.554141631	\\
9.7	2.559278662	\\
9.8	2.564363452	\\
9.9	2.569397051	\\
10	2.574380481	\\
10.1	2.579314732	\\
10.2	2.584200764	\\
10.3	2.589039511	\\
10.4	2.59383188	\\
10.5	2.598578752	\\
10.6	2.603280981	\\
10.7	2.607939399	\\
10.8	2.612554817	\\
10.9	2.617128019	\\
11	2.621659772	\\
11.1	2.626150821	\\
11.2	2.630601889	\\
11.3	2.635013682	\\
11.4	2.639386888	\\
11.5	2.643722175	\\
11.6	2.648020196	\\
11.7	2.652281585	\\
11.8	2.656506962	\\
11.9	2.660696931	\\
12	2.66485208	\\
12.1	2.668972982	\\
12.2	2.673060198	\\
12.3	2.677114274	\\
12.4	2.681135744	\\
12.5	2.685125127	\\
12.6	2.689082931	\\
12.7	2.693009653	\\
12.8	2.696905776	\\
12.9	2.700771775	\\
13	2.704608111	\\
13.1	2.708415236	\\
13.2	2.712193591	\\
13.3	2.715943609	\\
13.4	2.719665711	\\
13.5	2.723360309	\\
13.6	2.727027807	\\
13.7	2.7306686	\\
13.8	2.734283073	\\
13.9	2.737871605	\\
14	2.741434566	\\
14.1	2.744972316	\\
14.2	2.748485211	\\
14.3	2.751973598	\\
14.4	2.755437815	\\
14.5	2.758878195	\\
14.6	2.762295065	\\
14.7	2.765688743	\\
14.8	2.769059543	\\
14.9	2.772407769	\\
15	2.775733724	\\
15.1	2.779037701	\\
15.2	2.782319988	\\
15.3	2.785580869	\\
15.4	2.788820622	\\
15.5	2.792039517	\\
15.6	2.795237822	\\
15.7	2.798415799	\\
15.8	2.801573704	\\
15.9	2.80471179	\\
16	2.807830303	\\
16.1	2.810929486	\\
16.2	2.814009578	\\
16.3	2.817070812	\\
16.4	2.820113418	\\
16.5	2.823137621	\\
16.6	2.826143642	\\
16.7	2.829131698	\\
16.8	2.832102004	\\
16.9	2.835054769	\\
17	2.837990198	\\
17.1	2.840908494	\\
17.2	2.843809856	\\
17.3	2.84669448	\\
17.4	2.849562556	\\
17.5	2.852414275	\\
17.6	2.855249821	\\
17.7	2.858069378	\\
17.8	2.860873123	\\
17.9	2.863661235	\\
18	2.866433885	\\
18.1	2.869191245	\\
18.2	2.871933482	\\
18.3	2.874660762	\\
18.4	2.877373246	\\
18.5	2.880071094	\\
18.6	2.882754464	\\
18.7	2.885423509	\\
18.8	2.888078383	\\
18.9	2.890719234	\\
19	2.89334621	\\
19.1	2.895959456	\\
19.2	2.898559116	\\
19.3	2.901145328	\\
19.4	2.903718233	\\
19.5	2.906277965	\\
19.6	2.90882466	\\
19.7	2.911358449	\\
19.8	2.913879463	\\
19.9	2.91638783	\\
20	2.918883675	\\
	};
	\addlegendentry{ME Bound}
	
	\end{axis}
	\end{tikzpicture}
	\caption{Plots of the entropy of a Poisson distributed random variable and its Gaussian maximum entropy (ME) upper bound in \eqref{Poisson_UpperBound} for $\lambda \in [0,10]$.}\label{fig:PoissonGap}
\end{figure}
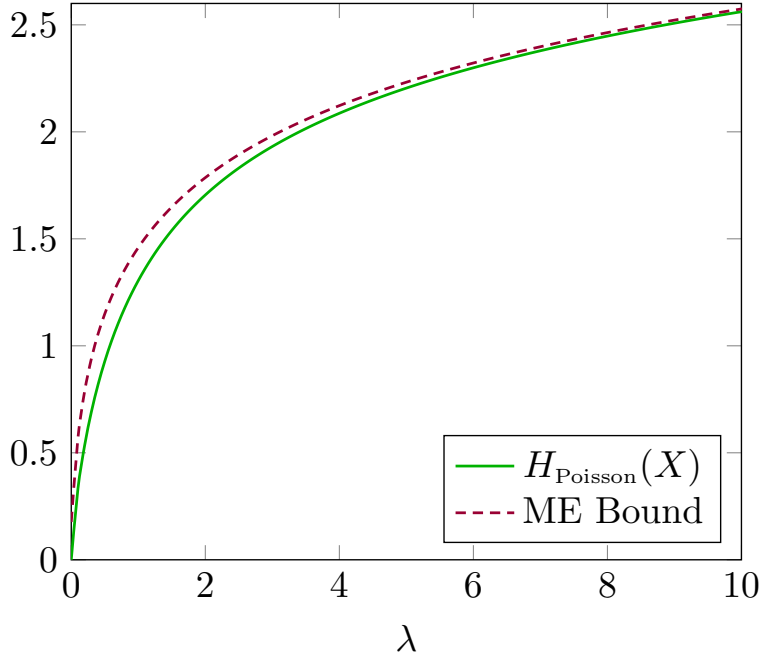

Since the upper bound \eqref{Univariate_Bound} in the univariate case seems to be relatively tight, at least in some cases, it seems very natural to generalize it and propose an upper bound on the entropy rate of discrete-valued stationary stochastic processes, which merely depends on the second-order statistics of the process.
Hence, the first objective of this work is to propose an upper bound on the entropy rates of discrete-valued stochastic processes in term of their PSD functions. 
In the same spirit, our second objective is to propose upper bounds on the entropy rate in terms of some finite collection of the covariance function of the stochastic process, i.e., bounds that merely depend on the set of second-order statistics $\{R_{\mbox{\tiny X}}(m)\}_{m=1}^{k}$ for some finite $k$.

\section{Main Results} \label{SEC4}

\subsection {Bounds via Gaussian Maximum-Entropy Principle}

The following result, which is proved in Appendix A, exhibits a generalization of the bound in \eqref{Univariate_Bound} to discrete-valued stationary stochastic processes. 

\begin{theorem} \label{Theorem_Gaussian}
	Let $\{Y_{n}\}_{n \geq 1}$ be a stationary process with a power spectral density function $\Phi_{\mbox{\tiny Y}}(\lambda)$.
	The entropy-rate of the process $\{Y_{n}\}_{n\geq1}$ is upper-bounded as
	\begin{align} \label{Gaussian_Upper_Bound}
	\bar{H}(Y)  
	\leq \frac{1}{2} \log (2 \pi e) 
	+ \frac{1}{4\pi} \int_{0}^{2\pi} \log \left(\Phi_{\mbox{\tiny Y}}(\lambda) + \frac{1}{12}\right) \dint\lambda.
	\end{align}
\end{theorem} 

In many cases of interest, where the statistical structure may be quite complicated, second order statistics can still be evaluated in closed form and the bound in \eqref{Gaussian_Upper_Bound} may be beneficial.    

%\subsubsection*{Moving-Average Processes}
As a first example, consider the DMA model, which is formed by taking probabilistic mixtures of independent identically distributed (i.i.d.) discrete random variables \cite{Jacobs}.   
In order to specify the model, we need first the following definitions. 
Let $\{Y_{n}\}$ be a sequence of i.i.d.\ random variables taking values in some countable subset $\calF$ of the real line, with $\prob\{Y_{n}=i\} = \pi_{i}$ for all $i \in \calF$. Let $\{\Theta_{n}\}$ be a sequence of i.i.d.\ random variables taking values in the set $\{0,1,\ldots,L\}$ with PMF
\begin{align}
\prob\{\Theta_{n}=\ell\} = \delta_{\ell},~~\ell=0,1,\ldots,L,
\end{align} 
where $L \in \natur$.
The DMA($L$) process $\{S_{n}\}$ is then formed by
\begin{align}
S_{n}=Y_{n-\Theta_{n}},
\end{align}
i.e., $S_{n}$ is a probabilistic mixture of the $L+1$ i.i.d.\ random variables $Y_{n}, Y_{n-1}, \ldots, Y_{n-L}$.
The covariance function of the DMA($L$) process is given by 
\begin{align}
R_{\mbox{\tiny S}}(k) 
= \left\{   
\begin{array}{l l}
\sigma_{\mbox{\tiny Y}}^{2}  & \quad k=0   \\
\sigma_{\mbox{\tiny Y}}^{2} \sum_{j=0}^{L-|k|} \delta_{j}\delta_{j+|k|}   & \quad 1\leq|k|\leq L   \\
0                           & \quad |k|\geq L+1,   \\
\end{array} \right.  
\end{align}
where $\sigma_{\mbox{\tiny Y}}^{2}$ denotes the variance of $Y_{n}$. In this case, the PSD function is given by 
\begin{align}
\Phi_{\mbox{\tiny S}}(\lambda) 
= \sigma_{\mbox{\tiny Y}}^{2} \left[1 + 2\sum_{k=1}^{L} \left(\sum_{j=0}^{L-k} \delta_{j}\delta_{j+k}\right) \cos(\lambda k)\right],
\end{align} 
and thus, the entropy rate of $\{S_{n}\}$ is upper-bounded by the result of substituting $\Phi_{\mbox{\tiny S}}(\lambda)$ into \eqref{Gaussian_Upper_Bound}. 
It is important to note that DMA($L$) is not in general a Markov process, so the standard conditional entropy bounds are not tight.   

For a second example, consider the following MA model of order $m$:
\begin{align} \label{General_MA_Process}
X_{n} = \sum_{i=1}^{m} \theta_{i} U_{n-i} + U_{n}, 
\end{align}  
where $\theta_{1},\ldots,\theta_{m}$ are the parameters of the model and $\{U_{n}\}_{n \in \natur}$ are i.i.d.\ standard Gaussians. Although this is a continuous-valued process in general, we may rely on the reasonable assumption that every process in nature is sampled and quantized to some finite precision in the first place, which obviously yields a discrete-valued process. For the original process in \eqref{General_MA_Process}, its covariance function $R_{\mbox{\tiny X}}(k)$ equals to zero as long as $k \geq m+1$, and this fact remains true also for its quantized version. We elaborate more on the quantized MA process in Subsection \ref{SEC_Example1}.         

%\subsubsection*{Hidden Markov Models}

We continue by referring to two specific HMPs.
The entropy rate of the HMP is a long standing open problem and a closed-form expression for it is not known, even for the simplest binary cases. A comprehensive survey on statistical properties of HMPs, mainly from the information-theoretic view point can be found in \cite{MERHAV2002}.   
The two simple models defined below potentially involve relatively large or even infinite alphabets, which makes trivial upper bounds on the entropy rate like $H(Y_{1}|Y_{0})$ less attractive, while the PSD functions of the resulting processes can be easily derived, as will be evident.    

Let $\{X_{n}\}_{n\geq1}$ be an irreducible homogeneous Markov chain on the state-space $\{1,2,\ldots,m\}$, with transition probability matrix $\bGamma$. That is, $\bGamma = (\gamma_{ij})$, where for all states $i,j$ and times $n$:
\begin{align}
\gamma_{ij} = \prob(X_{n}=j|X_{n-1}=i).
\end{align}
By the irreducibility of $\{X_{n}\}_{n\geq1}$, there exists a unique, strictly positive, stationary distribution, which we shall denote by the vector $\bdelta = (\delta_{1}, \delta_{2}, \ldots, \delta_{m})$. Suppose throughout that $\{X_{n}\}_{n\geq1}$ is stationary, so that $\bdelta$ is the distribution of $X_{n}$ for all $n$. Now let the nonnegative integer-valued random process $\{Y_{n}\}_{n\geq1}$ be such that, conditional on $\{X_{n}\}_{n=1}^{N_{0}}$, for any $N_{0} \in \natur$, the random variables $\{Y_{n}\}_{n=1}^{N_{0}}$ are mutually independent and, if $X_{n}=i$, $Y_{n}$ takes the value $s$ with probability $\pi_{si}$. That is, for $n=1,\ldots,N_{0}$, the distribution of $Y_{n}$ conditional on $\{X_{n}\}_{n=1}^{N_{0}}$ is given by
\begin{align}
\prob(Y_{n}=s|X_{n}=i) = \pi_{si}.
\end{align}    
We shall refer to the probabilities $\pi_{si}$ as the state-dependent probabilities. The two cases we shall refer to are: (i) the conditional distribution of $Y_{n}$ is binomial; and (ii) the conditional distribution of $Y_{n}$ is Poisson.  

In case (i), if $X_{n}=i$, then $Y_{n}$ has a binomial distribution with parameters $N$ and $p_{i} \in [0,1]$, and the state-dependent probabilities are given for all integers $s=0,1,\ldots,N$ by:
\begin{align}
\pi_{si} = \binom{N}{s} p_{i}^{s}(1-p_{i})^{N-s}.
\end{align}
In case (ii), if $X_{n}=i$, then $Y_{n}$ has a Poisson distribution with mean $\lambda_{i}\geq 0$, and the state-dependent probabilities are given for all nonnegative integers $s$ by:
\begin{align}
\pi_{si} = e^{-\lambda_{i}} \lambda_{i}^{s}/s!.
\end{align} 
We will refer to the models just defined as binomial-hidden and Poisson-hidden Markov processes, respectively, as was originally termed in \cite[Sec.\ 2.3]{MacDonald}.
In order to arrive at relatively simple close-form expressions, at least for one of these models, we confine ourselves to the binary case $m=2$. In this case, we write the transition probability matrix of $\{X_{n}\}_{n\geq1}$ as 
\begin{align}
\bGamma = 
\begin{pmatrix}
1-\gamma_{1}   & \gamma_{1}     \\
\gamma_{2}    & 1-\gamma_{2}        
\end{pmatrix},
\end{align}  
and it follows that 
\begin{align}
\bdelta =
\frac{1}{\gamma_{1}+\gamma_{2}} 
\begin{pmatrix}
\gamma_{2}   & \gamma_{1}          
\end{pmatrix}.
\end{align}
The following expression for $\bGamma^{k}$, obtained by diagonalizing $\bGamma$, will be useful in deriving statistical properties of $\{Y_{n}\}_{n\geq1}$ when $m=2$:
\begin{align}
\bGamma^{k} = 
\begin{pmatrix}
\gamma_{1} & \gamma_{2}     \\
\gamma_{1} & \gamma_{2}        
\end{pmatrix} + \omega^{k} 
\begin{pmatrix}
\gamma_{2}  & -\gamma_{2}     \\
-\gamma_{1} & \gamma_{1}        
\end{pmatrix},
\end{align}  
where $\omega=1-\gamma_{1}-\gamma_{2}$. As a preliminary for deriving the covariance function and the PSD function of the binomial-hidden Markov process, we state two useful results. First, provided the relevant expectations exist,
\begin{align} \label{First_order_statistics}
\Exp[f(Y_{n})] = \sum_{i=1}^{m} \Exp[f(Y_{n})|X_{n}=i] \delta_{i}.
\end{align} 
This is proved by conditioning on $X_{n}$ and noting that $\prob(X_{n}=i)=\delta_{i}$. Second, provided again that the relevant expectations exist, we have for $k \in \mathbb{N}$ that 
\begin{align} \label{Second_order_statistics}
\Exp[f(Y_{n},Y_{n+k})] = \sum_{i,j=1}^{m} \Exp[f(Y_{n},Y_{n+k})|X_{n}=i,X_{n+k}=j] \delta_{i} \gamma_{ij}(k),
\end{align}
where $\gamma_{ij}(k) = (\bGamma^{k})_{ij}$.
To prove this, we condition on $\bX^{n+k}=\{X_{\ell}:~\ell=1,\ldots,n+k\}$ and exploit the fact that the conditional expectation of $f(Y_{n},Y_{n+k})$ given $\bX^{n+k}$, is the conditional expectation given only $X_{n}$ and $X_{n+k}$. Summing $P(X_{1},\ldots,X_{n+k})$ over the states at all times other than $n$ and $n+k$ gives $P(X_{n},X_{n+k})$, and \eqref{Second_order_statistics} follows since $\prob(X_{n}=i,X_{n+k}=j)=\delta_{i} \gamma_{ij}(k)$. 

We now turn to derive the covariance function of the binomial-hidden Markov process.
We use the notation $\bp=(p_{1}, p_{2}, \ldots, p_{m})$ and $\bP = \text{diag}(\bp)$. 
From \eqref{First_order_statistics} we have     
\begin{align}
\Exp[Y_{n}] = \sum_{i=1}^{m} (Np_{i}) \delta_{i} = N \bdelta \bp'
\end{align}
and 
\begin{align}
\Exp[Y_{n}^{2}] 
&= \sum_{i=1}^{m} (Np_{i}(1-p_{i}) + N^{2}p_{i}^{2}) \delta_{i} \\
&= N \bdelta \bp' + N(N-1)\bdelta \bP \bp'.
\end{align}
Hence
\begin{align}
\text{Var}(Y_{n})
&= N \bdelta \bp' + N(N-1)\bdelta \bP \bp' - N^{2}(\bdelta \bp')^{2} \\
\label{ref1aa}
&= N^{2}(\bdelta \bP \bp' - (\bdelta \bp')^{2}) + N(\bdelta \bp' - \bdelta \bP \bp').
\end{align}
From \eqref{Second_order_statistics} we have for any $k \in \mathbb{N}$ that
\begin{align}
\Exp[Y_{n}Y_{n+k}] = \sum_{i,j=1}^{m} (Np_{i})(Np_{j}) \delta_{i} \gamma_{ij}(k) = N^{2} \bdelta \bP \bGamma^{k} \bp'.
\end{align}
Thus the resulted covariance is 
\begin{align}
\label{ref2aa}
\text{Cov}(Y_{n},Y_{n+k}) = N^{2}(\bdelta \bP \bGamma^{k} \bp' - (\bdelta \bp')^{2}).
\end{align}
In the binary case $m=2$, i.e., if the Markov chain has only two states, the bracketed terms in \eqref{ref1aa} and \eqref{ref2aa} are given by
\begin{align}
\bdelta \bP \bp' - (\bdelta \bp')^{2} &= \delta_{1} \delta_{2} (p_{2}-p_{1})^{2} \dfn \alpha \\
\bdelta \bP \bGamma^{k} \bp' - (\bdelta \bp')^{2}
&= \delta_{1} \delta_{2} (p_{2}-p_{1})^{2} \omega^{k}, 
\end{align} 
where $\omega=1-\gamma_{1}-\gamma_{2}$, and
\begin{align}
\bdelta \bp' - \bdelta \bP \bp' = \delta_{1} p_{1}(1-p_{1}) + \delta_{2} p_{2}(1-p_{2}) \dfn \beta.
\end{align}
Hence, the covariance function of the binomial-hidden Markov process is given by
\begin{align}
R_{\mbox{\tiny Y}}(k)
= \left\{   
\begin{array}{l l}
N^{2} \alpha + N \beta    & \quad k=0   \\
N^{2} \alpha \omega^{|k|}  & \quad k\neq 0,   \\
\end{array} \right. 
\end{align}
and its PSD function is
\begin{align}
\Phi_{\mbox{\tiny Y}}(\lambda)
%%%%%%%%%%%%%%%%%%%%%%%%%%%%%%%%%%%%%%%%
= N^{2}\alpha \cdot \frac{1 - \omega^{2}}{1 + \omega^{2} - 2\omega \cos(\lambda)} + N\beta.
\end{align}
%\begin{align}
%\Phi_{\mbox{\tiny Y}}(\lambda)
%&= \sum_{k=-\infty}^{\infty} R_{\mbox{\tiny Y}}(k) e^{-j \lambda k} \\
%%%%%%%%%%%%%%%%%%%%%%%%%%%%%%%%%%%%%%%%%
%&= \sum_{k=-\infty}^{0} N^{2} A \omega^{|k|} e^{-j \lambda k} + \sum_{k=0}^{\infty} N^{2} A \omega^{|k|} e^{-j \lambda k} + (NB-N^{2}A) \\
%%%%%%%%%%%%%%%%%%%%%%%%%%%%%%%%%%%%%%%%%
%&= N^{2} A \sum_{k=0}^{\infty} \omega^{k} e^{j \lambda k} + N^{2} A \sum_{k=0}^{\infty} \omega^{k} e^{-j \lambda k} + (NB-N^{2}A) \\
%%%%%%%%%%%%%%%%%%%%%%%%%%%%%%%%%%%%%%%%%
%&= \frac{N^{2}A}{1 - \omega e^{j\lambda}} + \frac{N^{2}A}{1 - \omega e^{-j\lambda}} + (NB-N^{2}A) \\
%%%%%%%%%%%%%%%%%%%%%%%%%%%%%%%%%%%%%%%%%
%&= N^{2}A \cdot \frac{2-2\omega \cos(\lambda)}{1 + \omega^{2} - 2\omega \cos(\lambda)} + (NB-N^{2}A) \\
%%%%%%%%%%%%%%%%%%%%%%%%%%%%%%%%%%%%%%%%%
%&= N^{2}A \cdot \frac{1 - \omega^{2}}{1 + \omega^{2} - 2\omega \cos(\lambda)} + NB.
%\end{align}
Finally, the entropy rate of the binomial-hidden Markov process is upper-bounded as
\begin{align} \label{Upper_Bound_BHMM}
\bar{H}(Y) 
&\leq \frac{1}{2} \log (2 \pi e) 
+ \frac{1}{4\pi} \int_{0}^{2\pi} \log \left(\frac{N^{2}\alpha(1 - \omega^{2})}{1 + \omega^{2} - 2\omega \cos(\lambda)} + N\beta + \frac{1}{12}\right) \dint\lambda,
\end{align}
where the last integral may be evaluated numerically to any degree of precision. Deriving an upper bound on the entropy rate of the Poisson-hidden Markov process can be done in a very similar fashion.

Nonetheless, the bound given in Theorem \ref{Theorem_Gaussian} has at least one major drawback, which is the requirement of a complete knowledge of the PSD function of the process in question. Although the PSD function is well defined for any stationary process, it is not always possible to calculate it in closed form, mainly because its calculation requires knowing the entire covariance function, and this alone may be quite demanding in some cases (e.g., see the application in Subsection \ref{SEC_Example2}). 
Another possible scenario is when the entire covariance function is given, but the PSD function, which is its discrete-time Fourier transform, as defined in \eqref{Def_PSD}, cannot be calculated in closed-form. Such cases are somewhat less common in real-life models.   
Hence, upper bounds on the entropy rates in such cases, where one has only partial knowledge on the covariance function, may be quite beneficial. Such bounds are presented in the section to follow.

\subsection{Bounds via Gibbs' Inequality}

Suppose that $p(\by)$ and $q(\by)$ are two probability density functions on $\calS \subseteq \reals^{n}$. 
It is well known \cite[Theorem 8.6.1]{Cover} that the Kullback-Leibler divergence between the densities $p$ and $q$ is always non-negative. Then, it holds that 
\begin{align} \label{Gibbs}
- \int_{\calS} p(\by) \log p(\by) \dint\by
\leq - \int_{\calS} p(\by) \log q(\by) \dint\by,
\end{align}
which is an upper bound on the differential entropy of $p$. Hence, in light of the fact that $H(\bY_{n}) = h(\tilde{\bY}_{n})$ (Lemma \ref{Basic_Lemma} above), any distribution $q(\by)$ can potentially yield an upper bound on $H(\bY_{n})$. 
Two bounds are derived via this machinery, relying on the multivariate t-distribution. 
A random vector $\bY = (Y_{1},\ldots,Y_{n})$ follows a multivariate t-distribution if its density $q_{\mbox{\tiny t}}$ has the form   
\begin{align} \label{Student_t_distribution}
q_{\mbox{\tiny t}}(\by) 
= C_{n} |\bSigma|^{-1/2} \left[1+\frac{1}{\nu}(\by-\bmu)^{\mbox{\tiny T}}\bSigma^{-1}(\by-\bmu)\right]^{-(\nu+n)/2},
\end{align}
where $\bmu \in \reals^{n}$, $\bSigma$ is a positive-definite $n \times n$ matrix, $\nu$ is the degrees of freedom, and the normalizing factor is 
\begin{align}
\label{Normalizing}
C_{n} = \frac{\Gamma[(\nu+n)/2]}{\Gamma(\nu/2)\nu^{n/2}\pi^{n/2}},
\end{align}
where the Gamma function 
\begin{align}
\Gamma(s) = \int_{0}^{\infty} x^{s-1} e^{-x} \dint x,~~~s>0.
\end{align}
Using this multivariate distribution yields the following result, which is proved in Appendix B.

\begin{theorem} \label{Theorem_t_Distribution_Order_K}
	Let $\{Y_{n}\}_{n \geq 1}$ be a stationary process of covariance function $R_{\mbox{\tiny Y}}(k)$. The entropy-rate of the process $\{Y_{n}\}_{n\geq1}$ is upper-bounded as
	\begin{align}
	\bar{H}(Y) 
	\leq \inf_{\{\bbeta \in \mathbb{R}^{k}:~\sum_{m=1}^{k} |\beta_{m}| < 1 \}} \left\{ \frac{1}{2} \log \left(2\pi e \Sigma(\bbeta)\right)
	-\frac{1}{4\pi} \int_{0}^{2\pi} \log \Psi(\bbeta,\lambda) \dint\lambda \right\},
	\end{align}
	where,
	\begin{align}
	\label{Def_Sigma}
	\Sigma(\bbeta)= \left(R_{\mbox{\tiny Y}}(0) + \frac{1}{12}\right) + \sum_{m=1}^{k} \beta_{m} R_{\mbox{\tiny Y}}(m),
	\end{align}
	and 
	\begin{align}
	\label{Def_Psi}
	\Psi(\bbeta,\lambda) = 1 + \sum_{m=1}^{k} \beta_{m} \cos (m \lambda). 
	\end{align}
\end{theorem} 

In comparison to \eqref{Gaussian_Upper_Bound}, where we just had to plug-in the PSD function and evaluate a finite-interval one-dimensional integral (which numerically takes only a fraction of a second), here, the situation is somewhat different. The partial knowledge of the second order statistics of the process has its penalty; the single-letter upper bound is more complicated to calculate: in addition to a one-dimensional integral, we also need to solve a minimization problem. This, however, can be done using numerical optimization methods. When we are able to rely on more second order statistics of the process in question, the dimension of the minimization problem increases, and the resulted upper bound is tighter.       

In Subsection \ref{SEC_Example2} below, we provide an example for a simple discrete-valued stochastic process, where its covariance function is given by a relatively cumbersome expression (a double infinite sum over an infinite integral, as can be seen in \eqref{AR_Covariance}). 
Hence, the calculation of the bound in \eqref{Gaussian_Upper_Bound} may be extremely exhausting, since it involves the calculation of an integral over the PSD function, which, in turn, is given by an infinite sum (the discrete-time Fourier transform) over the covariance function. This practically prevents the use of Theorem \ref{Theorem_Gaussian}. 
In this case, we rely on Theorem \ref{Theorem_t_Distribution_Order_K} to derive upper bounds on the entropy rate. As will be seen in Subsection \ref{SEC_Example2}, relying on more statistics yields a tighter upper bound on the entropy rate.       

Since the numerical problem involved in the bound of Theorem \ref{Theorem_t_Distribution_Order_K} is more demanding than the one in the bound of Theorem \ref{Theorem_Gaussian}, mainly due to the need to perform both numerical integration as well as numerical optimization (the bound \eqref{Gaussian_Upper_Bound} requires numerical integration only), we also propose a lighter version of it, which is proved in Appendix C. Although the bound in Theorem \ref{Theorem_t_Distribution_Order_1} below relies only on the variance and the covariance between two consecutive variables in the process, this bound may still be beneficial when those are the only available statistics. 

\begin{theorem} \label{Theorem_t_Distribution_Order_1}
	The entropy-rate of a stationary process $\{Y_{n}\}_{n\geq1}$ is upper-bounded as
	\begin{align} 
	\bar{H}(Y) 
	\leq \inf_{s \in (-1,1)} \frac{1}{2} \log \left(4\pi e \frac{\left(R_{\mbox{\tiny Y}}(0) + \tfrac{1}{12}\right) + s R_{\mbox{\tiny Y}}(1)}{1 + \sqrt{1-s^{2}}}\right).
	\end{align}
\end{theorem} 

The bounds in Theorems \ref{Theorem_t_Distribution_Order_K} and \ref{Theorem_t_Distribution_Order_1} stems from Gibbs' inequality and the multivariate t-distribution. Still, many more upper bounds may be derived using similar techniques, where the main idea is to plug-in some nice probability distribution into the right-hand-side of \eqref{Gibbs} and then to calculate its expectation with respect to the real statistics of the process. Relatively beneficial may be the family of elliptical distributions. An elliptical distribution with a probability density function $q$ has the form\footnote{In a more general settings, the family of elliptical distributions are defined via the characteristic function \cite{Cambanis}. In our settings, we seek for probability distributions with well defined densities.}
\begin{align}
q(\by) = k \cdot g((\by-\bmu)^{\mbox{\tiny T}}\bSigma^{-1}(\by-\bmu)),
\end{align}  
where $k$ is the normalizing factor, $\bmu$ is the mean vector, and $\bSigma$ is a positive definite matrix which is proportional to the covariance matrix. This family includes, among others, the multivariate normal distribution, the multivariate t-distribution as in \eqref{Student_t_distribution}, and the multivariate Laplace distribution, just to name a few.
Another example, which does not belong to the elliptical family of distributions, and that may lead to more useful bounds on the entropy rate is the multivariate log-normal distribution. 
A random vector $\bY = (Y_{1},\ldots,Y_{n})$ follows a multivariate log-normal distribution if the density function of $\bY$ is defined by   
\begin{align}
q_{\mbox{\tiny log}}(\by) 
= (2\pi)^{-n/2} |\bSigma|^{-1/2} \left(\prod_{i=1}^{n} y_{i} \right)^{-1}
\exp\left\{-\tfrac{1}{2}(\log\by-\bmu)^{\mbox{\tiny T}}\bSigma^{-1}(\log\by-\bmu)\right\},
\end{align}
with $\log\by \dfn (\log y_{1},\ldots,\log y_{n})^{\mbox{\tiny T}}$, $y_{i}>0$, $i=1,\ldots,n$, $\bmu \in \reals^{n}$, and $\bSigma$ is a positive-definite $n \times n$ matrix. 
This multivariate distribution is still relatively easy to handle but the resulted bounds are somewhat more cumbersome than the bounds given in Theorem \ref{Theorem_t_Distribution_Order_K}, thus will not be presented here in details. 
%We will not elaborate more on it here.   

\section{Applications to Quantized Processes} \label{SEC_Applications}

\subsection{The Quantized Moving-Average Process} \label{SEC_Example1}

We start by demonstrating the usefulness of the upper bound in Theorem \ref{Theorem_Gaussian} by referring to a quantized version of a simple MA process.
Let $\{W_{n}\}_{n \in \integers}$ be white noise, i.e., i.i.d.\ with $W_{n} \sim \calN(0,\sigma^{2})$, and for $\theta \in \mathbb{R}$, let $\{X_{n}\}_{n \in \integers}$ be a first-order MA process, defined by
\begin{align}
X_{n} = W_{n} + \theta W_{n-1}.
\end{align}
This process is Gaussian and stationary for any $\theta \in \mathbb{R}$, and its covariance function is given by
\begin{align}
R_{\mbox{\tiny X}}(k)
= \left\{   
\begin{array}{l l}
\sigma^{2}(1 + \theta^{2})  & \quad k=0   \\
\sigma^{2} \theta           & \quad k= \pm 1   \\
0                           & \quad \text{else}   \\
\end{array} \right. .
\end{align}
The PSD function is
\begin{align}
\Phi_{\mbox{\tiny X}}(\lambda) 
%= \sum_{k=-\infty}^{\infty} R_{\mbox{\tiny X}}(k) e^{i \lambda k}
= \sigma^{2}(1 + \theta^{2}) + 2 \sigma^{2} \theta \cos(\lambda),
\end{align} 
and the differential entropy rate of the process $\{X_{n}\}_{n \in \integers}$ is given by 
\begin{align}
\bar{H}(X) 
&= \frac{1}{2} \log (2 \pi e) 
+ \frac{1}{4\pi} \int_{0}^{2\pi} \log \Phi_{\mbox{\tiny X}}(\lambda) \dint\lambda \\
%%%%%%%%%%%%%%%%%%%%%%%%%%%%%%%%%%%%%%%%%%%%%%%%%%
&= \frac{1}{2} \log (2 \pi e) 
+ \frac{1}{4\pi} \int_{0}^{2\pi} \log [\sigma^{2}(1 + \theta^{2}) + 2 \sigma^{2} \theta \cos(\lambda)] \dint\lambda \\
%%%%%%%%%%%%%%%%%%%%%%%%%%%%%%%%%%%%%%%%%%%%%%%%%%
&= \frac{1}{2} \log (2 \pi e \sigma^{2}) 
+ \frac{1}{4\pi} \int_{0}^{2\pi} \log (1 + \theta^{2} + 2 \theta \cos(\lambda)) \dint\lambda \\
%%%%%%%%%%%%%%%%%%%%%%%%%%%%%%%%%%%%%%%%%%%%%%%%%%
&= \frac{1}{2} \log (2 \pi e \sigma^{2}).
\end{align}
Define the quantization function $\mathsf{Q}(\cdot)$ by
\begin{align} \label{DEF_QUAN}
\mathsf{Q}(s) = \operatorname*{arg\,min}_{m \in \mathbb{Z}} |s-m|,
\end{align}
and define the quantized MA process $\{Y_{n}\}$ by $Y_{n} = \mathsf{Q}(X_{n})$ at any time $n$. In general, if $\{X_{n}\}_{n \geq 1}$ is strongly stationary and $Y_{n} = f(X_{n})$, for any function $f(\cdot)$, then $\{Y_{n}\}_{n \geq 1}$ is also strongly stationary. 
Hence, the quantized MA process is strongly stationary and its PSD function is well defined. We start by calculating the covariance function. For any $m \geq n+2$, note that $\mathsf{Q}(W_{m} + \theta W_{m-1})$ and $\mathsf{Q}(W_{n} + \theta W_{n-1})$ are independent, thus $R_{\mbox{\tiny Y}}(k)=0$ for any $|k|\geq 2$. For $R_{\mbox{\tiny Y}}(0)$, consider the following. Let $Y \sim \calN(0,\sigma_{0}^{2})$ and note that $\Exp[\mathsf{Q}(Y)]=0$. The second moment is given by   
\begin{align}
\Exp[\mathsf{Q}(Y)^{2}] 
&= \sum_{k=-\infty}^{\infty} k^{2} \prob \left\{Y \in [k-\tfrac{1}{2},k+\tfrac{1}{2}) \right\} \\
&= \sum_{k=1}^{\infty} 2k^{2} \prob \left\{Y \in [k-\tfrac{1}{2},k+\tfrac{1}{2}) \right\} \\
\label{Second_Moment_Quan_Gauss}
&= \sum_{k=1}^{\infty} 2k^{2} \left[\Phi\left(\frac{k+\frac{1}{2}}{\sigma_{0}}\right) - \Phi\left(\frac{k-\frac{1}{2}}{\sigma_{0}}\right)\right],
\end{align}
and hence
\begin{align}
R_{\mbox{\tiny Y}}(0)
&=\Exp[\mathsf{Q}(W_{n} + \theta W_{n-1})^{2}] \\
&= \sum_{k=1}^{\infty} 2k^{2} \left[\Phi\left(\frac{k+\frac{1}{2}}{\sigma\sqrt{1 + \theta^{2}}}\right) - \Phi\left(\frac{k-\frac{1}{2}}{\sigma\sqrt{1 + \theta^{2}}}\right)\right] \\
%%%%%%%%%%%%%%%%%%%%%%%%%%%%%%%%%%%%%%%%%%%%%%%%%%%%
&\dfn F(\sigma,\theta).
\end{align}
The term $R_{\mbox{\tiny Y}}(1)$ is given by
\begin{align}
R_{\mbox{\tiny Y}}(1)
=\Exp[\mathsf{Q}(W_{n+1} + \theta W_{n}) \mathsf{Q}(W_{n} + \theta W_{n-1})],
\end{align}
and note that conditioned on $W_{n}=s$, the independence of $W_{n-1}$ and $W_{n+1}$ implies that 
\begin{align} \label{ref6}
R_{\mbox{\tiny Y}}(1)
= \int_{-\infty}^{\infty} \Exp[\mathsf{Q}(W_{n+1} + \theta s)] \Exp[\mathsf{Q}(s + \theta W_{n-1})] \tfrac{1}{\sqrt{2 \pi \sigma^{2}}} e^{-s^{2}/2\sigma^{2}} \mathrm{d}s.
\end{align} 
For $Z \sim \calN(\mu_{0},\sigma_{0}^{2})$, we have that
\begin{align}
\Exp[\mathsf{Q}(Z)] 
&= \sum_{m=-\infty}^{\infty} m \prob \left\{Z \in [m-\tfrac{1}{2},m+\tfrac{1}{2}) \right\} \\
\label{ref2}
&= \sum_{m=-\infty}^{\infty} m \left[\Phi\left(\frac{m - \mu_{0}+\frac{1}{2}}{\sigma_{0}}\right) - \Phi\left(\frac{m - \mu_{0}-\frac{1}{2}}{\sigma_{0}}\right)\right],
\end{align}
and thus, the two expectations inside the integral in \eqref{ref6} are given by
\begin{align}
\Exp[\mathsf{Q}(W_{n+1} + \theta s)]
&= \sum_{k=-\infty}^{\infty} k \left[\Phi\left(\frac{k - \theta s+\frac{1}{2}}{\sigma}\right) - \Phi\left(\frac{k - \theta s-\frac{1}{2}}{\sigma}\right)\right] \\
\label{ref8}
&\dfn \sum_{k=-\infty}^{\infty} k I_{0}(k,s,\sigma,\theta)  \\
\Exp[\mathsf{Q}(\theta W_{n-1} + s)]
&= \sum_{\ell=-\infty}^{\infty} \ell \left[\Phi\left(\frac{\ell - s+\frac{1}{2}}{\theta \sigma}\right) - \Phi\left(\frac{\ell - s-\frac{1}{2}}{\theta \sigma}\right)\right] \\
\label{ref9}
&\dfn \sum_{\ell=-\infty}^{\infty} \ell J_{0}(\ell,s,\sigma,\theta).
\end{align}
Substituting \eqref{ref8} and \eqref{ref9} back into \eqref{ref6} yields
\begin{align} 
R_{\mbox{\tiny Y}}(1)
&= \int_{-\infty}^{\infty} \left(\sum_{k=-\infty}^{\infty} k I_{0}(k,s,\sigma,\theta)\right) \left(\sum_{\ell=-\infty}^{\infty} \ell J_{0}(\ell,s,\sigma,\theta)\right) \tfrac{1}{\sqrt{2 \pi \sigma^{2}}} e^{-s^{2}/2\sigma^{2}} \mathrm{d}s \\
%%%%%%%%%%%%%%%%%%%%%%%%%%%%%%%%%%%%%%%%%%%%%%%%%%%%
&= \sum_{k=-\infty}^{\infty} \sum_{\ell=-\infty}^{\infty} k \ell \int_{-\infty}^{\infty}  I_{0}(k,s,\sigma,\theta)  J_{0}(\ell,s,\sigma,\theta) \tfrac{1}{\sqrt{2 \pi \sigma^{2}}} e^{-s^{2}/2\sigma^{2}} \mathrm{d}s \\
%%%%%%%%%%%%%%%%%%%%%%%%%%%%%%%%%%%%%%%%%%%%%%%%%%%%
&\dfn G(\sigma,\theta).
\end{align}
Now, the PSD function is given by
\begin{align}
\Phi_{\mbox{\tiny Y}}(\lambda) = F(\sigma,\theta) + 2 G(\sigma,\theta) \cos(\lambda),
\end{align}
and according to Theorem \ref{Theorem_Gaussian}, the entropy rate of the process $\{Y_{n}\}$ is upper-bounded by 
\begin{align}
\bar{H}(Y) 
&\leq \frac{1}{2} \log (2 \pi e) 
+ \frac{1}{4\pi} \int_{0}^{2\pi} \log \left(\Phi_{\mbox{\tiny Y}}(\lambda) + \tfrac{1}{12}\right) \dint\lambda \\
&= \frac{1}{2} \log (2 \pi e) 
+ \frac{1}{4\pi} \int_{0}^{2\pi} \log \left(F(\sigma,\theta) + \tfrac{1}{12} + 2 G(\sigma,\theta) \cos(\lambda) \right) \dint\lambda \\
\label{ref13}
&= \frac{1}{2} \log\left[2 \pi e(F(\sigma,\theta) + \tfrac{1}{12})\right] 
+ \frac{1}{4\pi} \int_{0}^{2\pi} \log \left(1 + \frac{2 G(\sigma,\theta)}{F(\sigma,\theta) + \tfrac{1}{12}} \cos(\lambda) \right) \dint\lambda.
\end{align}
To evaluate the integral in \eqref{ref13}, we use the following result, which is proved in Appendix D.  
\begin{lemma} \label{Lemma_Poisson}
	For any $s \in [-1,1]$, it holds that
	\begin{align}
	\frac{1}{2\pi} \int_{0}^{2\pi} \log \left(1+ s \cos(\lambda)\right) \dint\lambda
	= \left\{   
	\begin{array}{l l}
	-\log \left(\frac{2-2\sqrt{1-s^{2}}}{s^{2}}\right)  & \quad s \in [-1,0)\cup(0,1]   \\
	0           & \quad s=0  \\
	\end{array} \right. .
	\end{align}
\end{lemma}
In order to use Lemma \ref{Lemma_Poisson}, we must first prove that the fraction 
\begin{align}
K(\sigma,\theta) \dfn \frac{2 G(\sigma,\theta)}{F(\sigma,\theta) + \tfrac{1}{12}}
\end{align}
takes values in the range $[-1,1]$. Since the expressions of $F(\sigma,\theta)$ and $G(\sigma,\theta)$ are rather cumbersome functions of $\sigma$ and $\theta$, we calculate them numerically and plot in Figure \ref{fig:Values_of_K} curves of $K(\sigma,\theta)$ for $\sigma=1$ and $\sigma=5$.
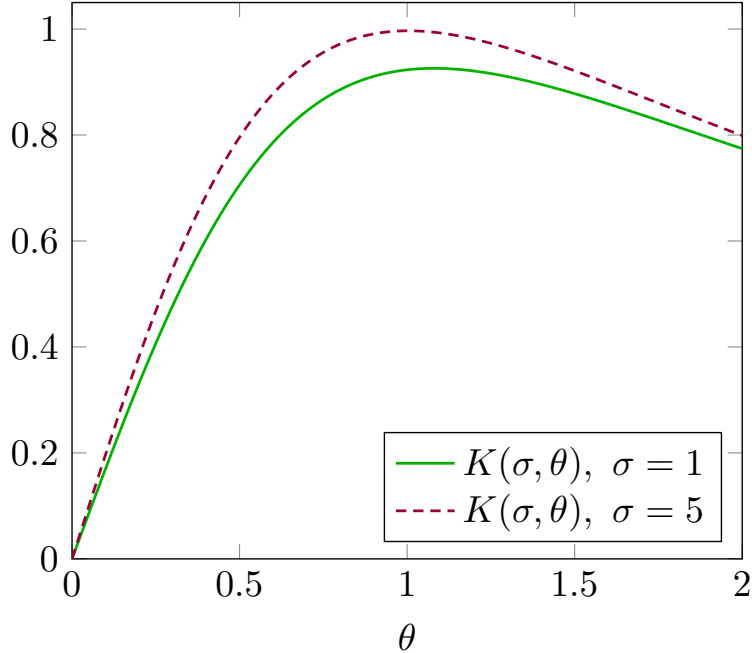
\begin{figure}[ht!]
	\centering
	\begin{tikzpicture}[scale=1.3]
	\begin{axis}[
	disabledatascaling,
	%x=25cm,
	%y=30cm,
	scaled x ticks=false,
	xticklabel style={/pgf/number format/fixed,
		/pgf/number format/precision=3},
	scaled y ticks=false,
	yticklabel style={/pgf/number format/fixed,
		/pgf/number format/precision=3},
	xlabel={$\theta$},
	xmin=0, xmax=2,
	ymin=0, ymax=1.05,
	legend pos=south east,
	%    ymajorgrids=true,
	%    grid style=dashed,
	]
	
	\addplot[smooth,color=black!30!green,thick]
	table[row sep=crcr] 
	{
0	0	\\
0.01	0.017141388	\\
0.02	0.034273963	\\
0.03	0.051388929	\\
0.04	0.068477517	\\
0.05	0.085531005	\\
0.06	0.102540731	\\
0.07	0.119498108	\\
0.08	0.136394635	\\
0.09	0.153221916	\\
0.1	0.169971671	\\
0.11	0.186635749	\\
0.12	0.203206141	\\
0.13	0.219674994	\\
0.14	0.236034618	\\
0.15	0.252277505	\\
0.16	0.268396332	\\
0.17	0.284383974	\\
0.18	0.300233515	\\
0.19	0.315938253	\\
0.2	0.331491713	\\
0.21	0.346887647	\\
0.22	0.362120048	\\
0.23	0.377183152	\\
0.24	0.392071444	\\
0.25	0.406779661	\\
0.26	0.421302798	\\
0.27	0.435636109	\\
0.28	0.449775112	\\
0.29	0.463715588	\\
0.3	0.477453581	\\
0.31	0.490985402	\\
0.32	0.504307628	\\
0.33	0.517417096	\\
0.34	0.530310908	\\
0.35	0.542986425	\\
0.36	0.555441267	\\
0.37	0.567673307	\\
0.38	0.579680667	\\
0.39	0.591461719	\\
0.4	0.603015075	\\
0.41	0.614339585	\\
0.42	0.625434329	\\
0.43	0.636298616	\\
0.44	0.646931974	\\
0.45	0.657334145	\\
0.46	0.667505079	\\
0.47	0.677444928	\\
0.48	0.687154037	\\
0.49	0.696632941	\\
0.5	0.705882353	\\
0.51	0.714903161	\\
0.52	0.723696419	\\
0.53	0.732263338	\\
0.54	0.740605285	\\
0.55	0.748723766	\\
0.56	0.756620429	\\
0.57	0.764297048	\\
0.58	0.771755522	\\
0.59	0.778997865	\\
0.6	0.786026201	\\
0.61	0.792842753	\\
0.62	0.799449841	\\
0.63	0.805849873	\\
0.64	0.812045339	\\
0.65	0.818038804	\\
0.66	0.823832903	\\
0.67	0.829430334	\\
0.68	0.834833852	\\
0.69	0.840046263	\\
0.7	0.845070423	\\
0.71	0.849909223	\\
0.72	0.854565596	\\
0.73	0.859042501	\\
0.74	0.863342926	\\
0.75	0.86746988	\\
0.76	0.871426387	\\
0.77	0.875215489	\\
0.78	0.878840231	\\
0.79	0.882303669	\\
0.8	0.885608856	\\
0.81	0.888758846	\\
0.82	0.891756688	\\
0.83	0.894605422	\\
0.84	0.897308076	\\
0.85	0.899867667	\\
0.86	0.902287193	\\
0.87	0.904569636	\\
0.88	0.906717956	\\
0.89	0.908735088	\\
0.9	0.910623946	\\
0.91	0.912387414	\\
0.92	0.914028348	\\
0.93	0.915549576	\\
0.94	0.916953892	\\
0.95	0.91824406	\\
0.96	0.919422807	\\
0.97	0.920492827	\\
0.98	0.921456779	\\
0.99	0.922317283	\\
1	0.923076923	\\
1.01	0.923738244	\\
1.02	0.924303752	\\
1.03	0.924775914	\\
1.04	0.925157158	\\
1.05	0.925449871	\\
1.06	0.925656401	\\
1.07	0.925779053	\\
1.08	0.925820094	\\
1.09	0.925781748	\\
1.1	0.925666199	\\
1.11	0.925475592	\\
1.12	0.925212028	\\
1.13	0.92487757	\\
1.14	0.924474239	\\
1.15	0.924004017	\\
1.16	0.923468846	\\
1.17	0.922870627	\\
1.18	0.922211223	\\
1.19	0.921492456	\\
1.2	0.920716113	\\
1.21	0.919883938	\\
1.22	0.91899764	\\
1.23	0.91805889	\\
1.24	0.917069323	\\
1.25	0.916030534	\\
1.26	0.914944087	\\
1.27	0.913811505	\\
1.28	0.912634281	\\
1.29	0.911413869	\\
1.3	0.910151692	\\
1.31	0.908849138	\\
1.32	0.907507563	\\
1.33	0.906128289	\\
1.34	0.904712607	\\
1.35	0.903261779	\\
1.36	0.901777031	\\
1.37	0.900259564	\\
1.38	0.898710546	\\
1.39	0.897131117	\\
1.4	0.895522388	\\
1.41	0.893885443	\\
1.42	0.892221338	\\
1.43	0.890531101	\\
1.44	0.888815735	\\
1.45	0.887076217	\\
1.46	0.885313498	\\
1.47	0.883528504	\\
1.48	0.881722138	\\
1.49	0.879895279	\\
1.5	0.87804878	\\
1.51	0.876183476	\\
1.52	0.874300176	\\
1.53	0.872399669	\\
1.54	0.870482722	\\
1.55	0.868550082	\\
1.56	0.866602474	\\
1.57	0.864640605	\\
1.58	0.862665162	\\
1.59	0.860676813	\\
1.6	0.858676208	\\
1.61	0.856663977	\\
1.62	0.854640734	\\
1.63	0.852607077	\\
1.64	0.850563585	\\
1.65	0.848510821	\\
1.66	0.846449332	\\
1.67	0.844379651	\\
1.68	0.842302293	\\
1.69	0.840217761	\\
1.7	0.838126541	\\
1.71	0.836029106	\\
1.72	0.833925916	\\
1.73	0.831817417	\\
1.74	0.82970404	\\
1.75	0.827586207	\\
1.76	0.825464324	\\
1.77	0.823338786	\\
1.78	0.821209978	\\
1.79	0.81907827	\\
1.8	0.816944024	\\
1.81	0.81480759	\\
1.82	0.812669306	\\
1.83	0.810529502	\\
1.84	0.808388495	\\
1.85	0.806246595	\\
1.86	0.804104101	\\
1.87	0.801961303	\\
1.88	0.799818481	\\
1.89	0.797675907	\\
1.9	0.795533845	\\
1.91	0.793392549	\\
1.92	0.791252267	\\
1.93	0.789113236	\\
1.94	0.786975688	\\
1.95	0.784839846	\\
1.96	0.782705926	\\
1.97	0.780574138	\\
1.98	0.778444683	\\
1.99	0.776317757	\\
2	0.774193548	\\
	};
	\legend{}
	\addlegendentry{$K(\sigma,\theta),~\sigma=1$}	
	
	\addplot[smooth,color=black!20!purple,thick,dash pattern={on 3pt off 2pt}]
	table[row sep=crcr]
	{
0	0	\\
0.01	0.019865576	\\
0.02	0.039719317	\\
0.03	0.059549409	\\
0.04	0.079344089	\\
0.05	0.09909166	\\
0.06	0.11878052	\\
0.07	0.138399183	\\
0.08	0.157936299	\\
0.09	0.177380679	\\
0.1	0.196721311	\\
0.11	0.215947387	\\
0.12	0.235048315	\\
0.13	0.254013743	\\
0.14	0.272833572	\\
0.15	0.291497976	\\
0.16	0.309997417	\\
0.17	0.328322657	\\
0.18	0.346464776	\\
0.19	0.364415178	\\
0.2	0.382165605	\\
0.21	0.39970815	\\
0.22	0.417035258	\\
0.23	0.434139743	\\
0.24	0.451014783	\\
0.25	0.467653936	\\
0.26	0.484051136	\\
0.27	0.500200698	\\
0.28	0.516097321	\\
0.29	0.531736088	\\
0.3	0.547112462	\\
0.31	0.562222289	\\
0.32	0.577061794	\\
0.33	0.591627573	\\
0.34	0.605916597	\\
0.35	0.619926199	\\
0.36	0.633654072	\\
0.37	0.64709826	\\
0.38	0.660257153	\\
0.39	0.673129477	\\
0.4	0.685714286	\\
0.41	0.698010953	\\
0.42	0.710019159	\\
0.43	0.721738887	\\
0.44	0.733170407	\\
0.45	0.744314266	\\
0.46	0.755171282	\\
0.47	0.765742526	\\
0.48	0.776029317	\\
0.49	0.786033206	\\
0.5	0.795755968	\\
0.51	0.80519959	\\
0.52	0.814366256	\\
0.53	0.823258343	\\
0.54	0.831878402	\\
0.55	0.840229153	\\
0.56	0.848313472	\\
0.57	0.856134378	\\
0.58	0.863695026	\\
0.59	0.870998696	\\
0.6	0.87804878	\\
0.61	0.884848778	\\
0.62	0.891402281	\\
0.63	0.897712969	\\
0.64	0.903784598	\\
0.65	0.909620991	\\
0.66	0.915226033	\\
0.67	0.92060366	\\
0.68	0.925757851	\\
0.69	0.930692624	\\
0.7	0.935412027	\\
0.71	0.939920129	\\
0.72	0.944221018	\\
0.73	0.948318791	\\
0.74	0.952217552	\\
0.75	0.955921402	\\
0.76	0.959434439	\\
0.77	0.962760748	\\
0.78	0.9659044	\\
0.79	0.968869448	\\
0.8	0.971659919	\\
0.81	0.974279815	\\
0.82	0.976733106	\\
0.83	0.979023729	\\
0.84	0.981155583	\\
0.85	0.98313253	\\
0.86	0.984958387	\\
0.87	0.986636929	\\
0.88	0.988171882	\\
0.89	0.989566925	\\
0.9	0.990825688	\\
0.91	0.991951747	\\
0.92	0.992948626	\\
0.93	0.993819794	\\
0.94	0.994568668	\\
0.95	0.995198603	\\
0.96	0.995712903	\\
0.97	0.99611481	\\
0.98	0.99640751	\\
0.99	0.996594131	\\
1	0.996677741	\\
1.01	0.996661349	\\
1.02	0.996547906	\\
1.03	0.996340303	\\
1.04	0.996041374	\\
1.05	0.995653892	\\
1.06	0.995180572	\\
1.07	0.994624072	\\
1.08	0.993986992	\\
1.09	0.993271874	\\
1.1	0.992481203	\\
1.11	0.991617408	\\
1.12	0.990682864	\\
1.13	0.989679887	\\
1.14	0.988610742	\\
1.15	0.987477639	\\
1.16	0.986282734	\\
1.17	0.985028134	\\
1.18	0.98371589	\\
1.19	0.982348004	\\
1.2	0.980926431	\\
1.21	0.979453071	\\
1.22	0.977929781	\\
1.23	0.976358369	\\
1.24	0.974740593	\\
1.25	0.973078171	\\
1.26	0.971372771	\\
1.27	0.96962602	\\
1.28	0.9678395	\\
1.29	0.966014752	\\
1.3	0.964153276	\\
1.31	0.962256528	\\
1.32	0.960325929	\\
1.33	0.958362857	\\
1.34	0.956368654	\\
1.35	0.954344624	\\
1.36	0.952292036	\\
1.37	0.950212122	\\
1.38	0.948106078	\\
1.39	0.945975069	\\
1.4	0.943820225	\\
1.41	0.941642643	\\
1.42	0.939443391	\\
1.43	0.937223503	\\
1.44	0.934983984	\\
1.45	0.932725811	\\
1.46	0.93044993	\\
1.47	0.928157261	\\
1.48	0.925848695	\\
1.49	0.923525097	\\
1.5	0.921187308	\\
1.51	0.918836141	\\
1.52	0.916472385	\\
1.53	0.914096807	\\
1.54	0.911710147	\\
1.55	0.909313126	\\
1.56	0.906906441	\\
1.57	0.904490768	\\
1.58	0.902066761	\\
1.59	0.899635054	\\
1.6	0.897196262	\\
1.61	0.89475098	\\
1.62	0.892299783	\\
1.63	0.889843231	\\
1.64	0.887381863	\\
1.65	0.884916201	\\
1.66	0.882446752	\\
1.67	0.879974005	\\
1.68	0.877498433	\\
1.69	0.875020495	\\
1.7	0.872540633	\\
1.71	0.870059276	\\
1.72	0.867576838	\\
1.73	0.865093718	\\
1.74	0.862610305	\\
1.75	0.860126971	\\
1.76	0.857644078	\\
1.77	0.855161973	\\
1.78	0.852680995	\\
1.79	0.850201468	\\
1.8	0.847723705	\\
1.81	0.845248009	\\
1.82	0.842774674	\\
1.83	0.840303979	\\
1.84	0.837836197	\\
1.85	0.83537159	\\
1.86	0.83291041	\\
1.87	0.8304529	\\
1.88	0.827999295	\\
1.89	0.825549821	\\
1.9	0.823104693	\\
1.91	0.820664122	\\
1.92	0.818228308	\\
1.93	0.815797446	\\
1.94	0.813371719	\\
1.95	0.810951308	\\
1.96	0.808536384	\\
1.97	0.806127112	\\
1.98	0.803723649	\\
1.99	0.801326148	\\
2	0.798934754	\\
	};
	\addlegendentry{$K(\sigma,\theta),~\sigma=5$}
	
	\end{axis}
	\end{tikzpicture}
	\caption{Plots of $K(\sigma,\theta)$ for $\theta \in [0,2]$ and two values of $\sigma$.}\label{fig:Values_of_K}
\end{figure}
As can be seen in Figure \ref{fig:Values_of_K}, the values of $K(\sigma,\theta)$ are limited to the range $[0,1]$, hence Lemma \ref{Lemma_Poisson} is applicable in this case and we conclude that for any $\theta \in (0,2]$, 
\begin{align}
\bar{H}(Y) 
&\leq \frac{1}{2} \log\left[2 \pi e(F(\sigma,\theta) + \tfrac{1}{12})\right] 
- \frac{1}{2}\log \left(\frac{2-2\sqrt{1-K(\sigma,\theta)^{2}}}{K(\sigma,\theta)^{2}}\right) \\
\label{Upper_Bound_Gauss}
&= \frac{1}{2} \log\left[\frac{\pi e(F(\sigma,\theta) + \tfrac{1}{12})K(\sigma,\theta)^{2}}{1-\sqrt{1-K(\sigma,\theta)^{2}}}\right], 
\end{align}
while for $\theta=0$, the process $\{Y_{n}\}$ is i.i.d.\ and its entropy is upper-bounded by
\begin{align}
\bar{H}(Y) 
&\leq \frac{1}{2} \log\left[2 \pi e(F(\sigma,0) + \tfrac{1}{12})\right],
\end{align}
which is merely \eqref{Univariate_Bound}.
Let us denote this bound for any $\theta \geq 0$ by the function $\mathsf{H}_{\mbox{\tiny TH-1}}(\sigma,\theta)$.
Regarding the trivial bound $H(Y_{n+1}|Y_{n})$, first note that the marginal distribution of $Y_{n}$ is  
\begin{align}
P_{Y_{n}}(i) = \Phi\left(\frac{i+\frac{1}{2}}{\sigma\sqrt{1 + \theta^{2}}}\right) - \Phi\left(\frac{i-\frac{1}{2}}{\sigma\sqrt{1 + \theta^{2}}}\right).
\end{align}
The joint distribution of $(Y_{n},Y_{n+1})$ is calculated as
\begin{align}
&P_{Y_{n}Y_{n+1}}(i,j) \nn \\
&= \prob \left\{Y_{n}=i,~Y_{n+1}=j \right\} \\
&= \prob \left\{i-\tfrac{1}{2} \leq X_{n} \leq i+\tfrac{1}{2},~j-\tfrac{1}{2} \leq X_{n+1} \leq j+\tfrac{1}{2} \right\} \\
&= \prob \left\{i-\tfrac{1}{2} \leq W_{n}+\theta W_{n-1} \leq i+\tfrac{1}{2},~j-\tfrac{1}{2} \leq W_{n+1}+\theta W_{n} \leq j+\tfrac{1}{2} \right\} \\
&= \int_{-\infty}^{\infty} \prob \left\{i-\tfrac{1}{2} \leq s+\theta W_{n-1} \leq i+\tfrac{1}{2},~j-\tfrac{1}{2} \leq W_{n+1}+\theta s \leq j+\tfrac{1}{2} \right\} \tfrac{1}{\sqrt{2 \pi \sigma^{2}}} e^{-s^{2}/2\sigma^{2}} \mathrm{d}s \\
&= \int_{-\infty}^{\infty} \prob \left\{i-\tfrac{1}{2} \leq s+\theta W_{n-1} \leq i+\tfrac{1}{2}\right\} \prob \left\{ j-\tfrac{1}{2} \leq W_{n+1}+\theta s \leq j+\tfrac{1}{2} \right\} \tfrac{1}{\sqrt{2 \pi \sigma^{2}}} e^{-s^{2}/2\sigma^{2}} \mathrm{d}s \\
&= \int_{-\infty}^{\infty} J_{0}(i,s,\sigma,\theta) I_{0}(j,s,\sigma,\theta) \tfrac{1}{\sqrt{2 \pi \sigma^{2}}} e^{-s^{2}/2\sigma^{2}} \mathrm{d}s,
\end{align}
where $I_{0}$ and $J_{0}$ are defined in \eqref{ref8} and \eqref{ref9}, respectively. 
Now, the conditional entropy is given by
\begin{align}
H(Y_{n+1}|Y_{n}) = - \sum_{i=-\infty}^{\infty} \sum_{j=-\infty}^{\infty} P_{Y_{n}Y_{n+1}}(i,j) \log \frac{P_{Y_{n}Y_{n+1}}(i,j)}{P_{Y_{n}}(i)},
\end{align}
which is merely a function of $\sigma$ and $\theta$, to be denoted by $\mathsf{H}_{\mbox{\tiny CE}}(\sigma,\theta)$.
We now compare numerically the bound $\mathsf{H}_{\mbox{\tiny TH-1}}(\sigma,\theta)$ and the trivial bound $H(Y_{n+1}|Y_{n})$ in the specific cases of $\sigma=1$ and $\sigma=5$.
As can be seen in Figure \ref{fig:Gauss}, in some intermediate range of $\theta$ values, the new upper bound from Theorem \ref{Theorem_Gaussian} outperforms the conditional entropy upper bound. Still, for relatively low values of $\theta$, the conditional entropy bound is lower than the bound from Theorem \ref{Theorem_Gaussian}, which is not very surprising, since at the extreme of $\theta=0$, the process $\{Y_{n}\}$ is i.i.d.\ and $H(Y_{n+1}|Y_{n})$ yields an exact estimation for the entropy rate.   

We also compare to the bounds that stems from Gibbs' inequality. Since $R_{\mbox{\tiny Y}}(k)=0$ for any $|k|\geq 2$ for the quantized MA process, and only $R_{\mbox{\tiny Y}}(0)$ and $R_{\mbox{\tiny Y}}(1)$ are non-negative, we calculate numerically the bound from Theorem \ref{Theorem_t_Distribution_Order_1} rather than the bound from Theorem \ref{Theorem_t_Distribution_Order_K}. We denote this bound by $\mathsf{H}_{\mbox{\tiny TH-3}}(\sigma,\theta)$. As can be seen in Figure \ref{fig:Gauss}, this bound is worse than the maximum between $\mathsf{H}_{\mbox{\tiny CE}}(\sigma,\theta)$ and $\mathsf{H}_{\mbox{\tiny TH-1}}(\sigma,\theta)$, hence useless in this case.     

\begin{figure}[h!]	
	\begin{subfigure}[b]{0.5\columnwidth}
		\centering 
			\begin{tikzpicture}[scale=1]
		\begin{axis}[
		disabledatascaling,
		%x=25cm,
		%y=30cm,
		scaled x ticks=false,
		xticklabel style={/pgf/number format/fixed,
			/pgf/number format/precision=3},
		scaled y ticks=false,
		yticklabel style={/pgf/number format/fixed,
			/pgf/number format/precision=3},
		xlabel={$\theta$},
		xmin=0, xmax=2,
		ymin=1.4, ymax=2.2,
		legend pos=north west,
		%    ymajorgrids=true,
		%    grid style=dashed,
		]
		
		\addplot[smooth,color=black!20!orange,thick,dash pattern={on 3pt off 2pt on 1pt off 2pt}]
		table[row sep=crcr] 
		{
			0	1.45895882	\\
			0.1	1.459352628	\\
			0.2	1.460981545	\\
			0.3	1.465057633	\\
			0.4	1.473206096	\\
			0.5	1.487003295	\\
			0.6	1.507554289	\\
			0.7	1.53524315	\\
			0.8	1.569712759	\\
			0.9	1.610032901	\\
			1	1.654953857	\\
			1.1	1.70314416	\\
			1.2	1.753354991	\\
			1.3	1.804501478	\\
			1.4	1.855679944	\\
			1.5	1.906148268	\\
			1.6	1.95529244	\\
			1.7	2.002594145	\\
			1.8	2.047606506	\\
			1.9	2.089939627	\\
			2	2.129254427	\\
		};
		\legend{}
		\addlegendentry{$\mathsf{H}_{\mbox{\tiny CE}}(\sigma,\theta)$}

		\addplot[smooth,color=black!30!green,thick]
		table[row sep=crcr]
		{
			0	1.496013873	\\
			0.01	1.496019996	\\
			0.02	1.49603837	\\
			0.03	1.496069009	\\
			0.04	1.496111938	\\
			0.05	1.496167192	\\
			0.06	1.496234817	\\
			0.07	1.496314866	\\
			0.08	1.496407405	\\
			0.09	1.496512509	\\
			0.1	1.496630264	\\
			0.11	1.496760766	\\
			0.12	1.496904123	\\
			0.13	1.497060452	\\
			0.14	1.497229882	\\
			0.15	1.497412555	\\
			0.16	1.497608621	\\
			0.17	1.497818244	\\
			0.18	1.4980416	\\
			0.19	1.498278877	\\
			0.2	1.498530275	\\
			0.21	1.498796009	\\
			0.22	1.499076306	\\
			0.23	1.499371405	\\
			0.24	1.499681562	\\
			0.25	1.500007045	\\
			0.26	1.500348141	\\
			0.27	1.500705147	\\
			0.28	1.50107838	\\
			0.29	1.501468172	\\
			0.3	1.501874872	\\
			0.31	1.502298847	\\
			0.32	1.502740482	\\
			0.33	1.50320018	\\
			0.34	1.503678365	\\
			0.35	1.504175479	\\
			0.36	1.504691989	\\
			0.37	1.505228378	\\
			0.38	1.505785156	\\
			0.39	1.506362855	\\
			0.4	1.506962029	\\
			0.41	1.50758326	\\
			0.42	1.508227153	\\
			0.43	1.508894342	\\
			0.44	1.509585487	\\
			0.45	1.510301278	\\
			0.46	1.511042433	\\
			0.47	1.511809702	\\
			0.48	1.512603866	\\
			0.49	1.513425739	\\
			0.5	1.514276167	\\
			0.51	1.515156033	\\
			0.52	1.516066254	\\
			0.53	1.517007784	\\
			0.54	1.517981615	\\
			0.55	1.518988777	\\
			0.56	1.520030339	\\
			0.57	1.521107411	\\
			0.58	1.522221144	\\
			0.59	1.52337273	\\
			0.6	1.524563402	\\
			0.61	1.525794436	\\
			0.62	1.527067153	\\
			0.63	1.528382912	\\
			0.64	1.529743119	\\
			0.65	1.531149218	\\
			0.66	1.532602697	\\
			0.67	1.534105083	\\
			0.68	1.535657941	\\
			0.69	1.537262875	\\
			0.7	1.538921522	\\
			0.71	1.540635553	\\
			0.72	1.542406665	\\
			0.73	1.544236582	\\
			0.74	1.54612705	\\
			0.75	1.54807983	\\
			0.76	1.550096692	\\
			0.77	1.552179412	\\
			0.78	1.554329763	\\
			0.79	1.556549506	\\
			0.8	1.558840386	\\
			0.81	1.561204118	\\
			0.82	1.563642379	\\
			0.83	1.566156802	\\
			0.84	1.568748956	\\
			0.85	1.571420345	\\
			0.86	1.574172388	\\
			0.87	1.577006412	\\
			0.88	1.579923638	\\
			0.89	1.582925168	\\
			0.9	1.586011976	\\
			0.91	1.589184894	\\
			0.92	1.592444602	\\
			0.93	1.595791618	\\
			0.94	1.599226287	\\
			0.95	1.602748777	\\
			0.96	1.606359069	\\
			0.97	1.61005695	\\
			0.98	1.613842016	\\
			0.99	1.617713661	\\
			1	1.621671087	\\
			1.01	1.625713297	\\
			1.02	1.629839103	\\
			1.03	1.63404713	\\
			1.04	1.638335825	\\
			1.05	1.642703461	\\
			1.06	1.647148153	\\
			1.07	1.651667866	\\
			1.08	1.656260428	\\
			1.09	1.660923541	\\
			1.1	1.665654801	\\
			1.11	1.670451706	\\
			1.12	1.675311676	\\
			1.13	1.680232063	\\
			1.14	1.685210169	\\
			1.15	1.690243258	\\
			1.16	1.695328572	\\
			1.17	1.700463341	\\
			1.18	1.705644797	\\
			1.19	1.710870186	\\
			1.2	1.716136776	\\
			1.21	1.721441872	\\
			1.22	1.726782817	\\
			1.23	1.732157004	\\
			1.24	1.737561883	\\
			1.25	1.742994968	\\
			1.26	1.748453834	\\
			1.27	1.753936133	\\
			1.28	1.759439587	\\
			1.29	1.764961994	\\
			1.3	1.770501233	\\
			1.31	1.77605526	\\
			1.32	1.781622112	\\
			1.33	1.787199904	\\
			1.34	1.792786834	\\
			1.35	1.798381175	\\
			1.36	1.80398128	\\
			1.37	1.809585579	\\
			1.38	1.815192575	\\
			1.39	1.820800845	\\
			1.4	1.826409037	\\
			1.41	1.832015869	\\
			1.42	1.837620125	\\
			1.43	1.843220652	\\
			1.44	1.848816362	\\
			1.45	1.854406226	\\
			1.46	1.859989272	\\
			1.47	1.865564584	\\
			1.48	1.8711313	\\
			1.49	1.876688606	\\
			1.5	1.88223574	\\
			1.51	1.887771984	\\
			1.52	1.893296664	\\
			1.53	1.898809151	\\
			1.54	1.904308852	\\
			1.55	1.909795217	\\
			1.56	1.915267728	\\
			1.57	1.920725904	\\
			1.58	1.926169296	\\
			1.59	1.931597488	\\
			1.6	1.93701009	\\
			1.61	1.942406744	\\
			1.62	1.947787115	\\
			1.63	1.953150897	\\
			1.64	1.958497804	\\
			1.65	1.963827576	\\
			1.66	1.969139972	\\
			1.67	1.974434772	\\
			1.68	1.979711775	\\
			1.69	1.9849708	\\
			1.7	1.99021168	\\
			1.71	1.995434266	\\
			1.72	2.000638425	\\
			1.73	2.005824036	\\
			1.74	2.010990994	\\
			1.75	2.016139205	\\
			1.76	2.02126859	\\
			1.77	2.026379077	\\
			1.78	2.03147061	\\
			1.79	2.036543138	\\
			1.8	2.041596624	\\
			1.81	2.046631038	\\
			1.82	2.051646358	\\
			1.83	2.056642571	\\
			1.84	2.061619672	\\
			1.85	2.066577662	\\
			1.86	2.071516549	\\
			1.87	2.076436347	\\
			1.88	2.081337078	\\
			1.89	2.086218766	\\
			1.9	2.091081443	\\
			1.91	2.095925145	\\
			1.92	2.100749912	\\
			1.93	2.10555579	\\
			1.94	2.110342827	\\
			1.95	2.115111076	\\
			1.96	2.119860594	\\
			1.97	2.124591439	\\
			1.98	2.129303674	\\
			1.99	2.133997366	\\
			2	2.138672581	\\
		};
		\addlegendentry{$\mathsf{H}_{\mbox{\tiny TH-1}}(\sigma,\theta)$}

		\addplot[smooth,color=black!20!blue,thick,dash pattern={on 4pt off 2pt}]
		table[row sep=crcr]
		{
0	1.496013873	\\
0.1	1.496656924	\\
0.2	1.498941356	\\
0.3	1.503830439	\\
0.4	1.512619471	\\
0.5	1.526573281	\\
0.6	1.546592402	\\
0.7	1.573008724	\\
0.8	1.605555174	\\
0.9	1.643483604	\\
1	1.685758684	\\
1.1	1.731252308	\\
1.2	1.778890778	\\
1.3	1.827740938	\\
1.4	1.877044097	\\
1.5	1.926215005	\\
1.6	1.974822174	\\
1.7	2.022561121	\\
1.8	2.069227385	\\
1.9	2.114692595	\\
2	2.158884744	\\
		};
		\addlegendentry{$\mathsf{H}_{\mbox{\tiny TH-3}}(\sigma,\theta)$}
		
		\end{axis}
		\end{tikzpicture}
		\caption{For $\sigma=1$} 
		\label{fig:POPULATION1}
	\end{subfigure}% 
	\begin{subfigure}[b]{0.5\columnwidth}
		\centering 
	\begin{tikzpicture}[scale=1]
\begin{axis}[
disabledatascaling,
%x=25cm,
%y=30cm,
scaled x ticks=false,
xticklabel style={/pgf/number format/fixed,
	/pgf/number format/precision=3},
scaled y ticks=false,
yticklabel style={/pgf/number format/fixed,
	/pgf/number format/precision=3},
xlabel={$\theta$},
xmin=0, xmax=2,
ymin=3, ymax=3.8,
legend pos=north west,
%    ymajorgrids=true,
%    grid style=dashed,
]

\addplot[smooth,color=black!20!orange,thick,dash pattern={on 3pt off 2pt on 1pt off 2pt}]
table[row sep=crcr] 
{
	0	3.030040341	\\
	0.1	3.030105839	\\
	0.2	3.030867458	\\
	0.3	3.033854229	\\
	0.4	3.041111053	\\
	0.5	3.05460708	\\
	0.6	3.075696688	\\
	0.7	3.104811123	\\
	0.8	3.141453621	\\
	0.9	3.184432481	\\
	1	3.232184454	\\
	1.1	3.283054177	\\
	1.2	3.33546656	\\
	1.3	3.387996987	\\
	1.4	3.439377473	\\
	1.5	3.488478489	\\
	1.6	3.534292265	\\
	1.7	3.575927672	\\
	1.8	3.612615983	\\
	1.9	3.643721893	\\
	2	3.668753836	\\
};
\legend{}
\addlegendentry{$\mathsf{H}_{\mbox{\tiny CE}}(\sigma,\theta)$}

\addplot[smooth,color=black!30!green,thick]
table[row sep=crcr]
{
	0	3.031698717	\\
	0.01	3.031699046	\\
	0.02	3.031700033	\\
	0.03	3.03170168	\\
	0.04	3.031703988	\\
	0.05	3.031706961	\\
	0.06	3.031710601	\\
	0.07	3.031714913	\\
	0.08	3.031719902	\\
	0.09	3.031725575	\\
	0.1	3.031731937	\\
	0.11	3.031738997	\\
	0.12	3.031746764	\\
	0.13	3.031755246	\\
	0.14	3.031764455	\\
	0.15	3.031774401	\\
	0.16	3.031785098	\\
	0.17	3.031796559	\\
	0.18	3.031808798	\\
	0.19	3.031821832	\\
	0.2	3.031835677	\\
	0.21	3.031850351	\\
	0.22	3.031865875	\\
	0.23	3.031882269	\\
	0.24	3.031899555	\\
	0.25	3.031917759	\\
	0.26	3.031936904	\\
	0.27	3.031957019	\\
	0.28	3.031978133	\\
	0.29	3.032000278	\\
	0.3	3.032023485	\\
	0.31	3.032047791	\\
	0.32	3.032073234	\\
	0.33	3.032099854	\\
	0.34	3.032127693	\\
	0.35	3.032156799	\\
	0.36	3.032187219	\\
	0.37	3.032219006	\\
	0.38	3.032252216	\\
	0.39	3.032286909	\\
	0.4	3.032323148	\\
	0.41	3.032361002	\\
	0.42	3.032400544	\\
	0.43	3.032441853	\\
	0.44	3.032485012	\\
	0.45	3.032530111	\\
	0.46	3.032577249	\\
	0.47	3.032626529	\\
	0.48	3.032678064	\\
	0.49	3.032731974	\\
	0.5	3.032788391	\\
	0.51	3.032847455	\\
	0.52	3.03290932	\\
	0.53	3.032974151	\\
	0.54	3.033042127	\\
	0.55	3.033113444	\\
	0.56	3.033188315	\\
	0.57	3.033266969	\\
	0.58	3.033349661	\\
	0.59	3.033436666	\\
	0.6	3.033528287	\\
	0.61	3.033624856	\\
	0.62	3.033726739	\\
	0.63	3.03383434	\\
	0.64	3.033948102	\\
	0.65	3.034068521	\\
	0.66	3.034196142	\\
	0.67	3.034331575	\\
	0.68	3.034475501	\\
	0.69	3.03462868	\\
	0.7	3.034791967	\\
	0.71	3.034966323	\\
	0.72	3.035152833	\\
	0.73	3.03535273	\\
	0.74	3.035567412	\\
	0.75	3.035798479	\\
	0.76	3.036047762	\\
	0.77	3.036317369	\\
	0.78	3.036609736	\\
	0.79	3.036927692	\\
	0.8	3.037274535	\\
	0.81	3.037654133	\\
	0.82	3.038071039	\\
	0.83	3.038530645	\\
	0.84	3.03903937	\\
	0.85	3.039604889	\\
	0.86	3.040236435	\\
	0.87	3.040945162	\\
	0.88	3.041744603	\\
	0.89	3.042651239	\\
	0.9	3.043685178	\\
	0.91	3.044870949	\\
	0.92	3.046238391	\\
	0.93	3.047823511	\\
	0.94	3.049669145	\\
	0.95	3.051825048	\\
	0.96	3.054346853	\\
	0.97	3.057293251	\\
	0.98	3.060720809	\\
	0.99	3.064676595	\\
	1	3.069189943	\\
	1.01	3.074265951	\\
	1.02	3.079883373	\\
	1.03	3.085997956	\\
	1.04	3.092549843	\\
	1.05	3.09947221	\\
	1.06	3.10669862	\\
	1.07	3.114167883	\\
	1.08	3.121826473	\\
	1.09	3.129629144	\\
	1.1	3.137538492	\\
	1.11	3.145524022	\\
	1.12	3.153561067	\\
	1.13	3.161629756	\\
	1.14	3.169714107	\\
	1.15	3.177801265	\\
	1.16	3.185880884	\\
	1.17	3.193944622	\\
	1.18	3.201985748	\\
	1.19	3.20999883	\\
	1.2	3.217979478	\\
	1.21	3.225924151	\\
	1.22	3.233829997	\\
	1.23	3.241694727	\\
	1.24	3.249516516	\\
	1.25	3.257293919	\\
	1.26	3.265025805	\\
	1.27	3.272711304	\\
	1.28	3.280349766	\\
	1.29	3.287940721	\\
	1.3	3.295483849	\\
	1.31	3.302978961	\\
	1.32	3.310425972	\\
	1.33	3.31782489	\\
	1.34	3.325175795	\\
	1.35	3.332478835	\\
	1.36	3.339734209	\\
	1.37	3.346942161	\\
	1.38	3.354102976	\\
	1.39	3.361216969	\\
	1.4	3.368284479	\\
	1.41	3.375305873	\\
	1.42	3.382281531	\\
	1.43	3.38921185	\\
	1.44	3.396097241	\\
	1.45	3.402938121	\\
	1.46	3.409734917	\\
	1.47	3.416488063	\\
	1.48	3.423197993	\\
	1.49	3.429865149	\\
	1.5	3.436489971	\\
	1.51	3.443072902	\\
	1.52	3.449614383	\\
	1.53	3.456114855	\\
	1.54	3.462574757	\\
	1.55	3.468994528	\\
	1.56	3.475374601	\\
	1.57	3.481715409	\\
	1.58	3.488017379	\\
	1.59	3.494280936	\\
	1.6	3.500506502	\\
	1.61	3.506694492	\\
	1.62	3.512845319	\\
	1.63	3.51895939	\\
	1.64	3.525037109	\\
	1.65	3.531078874	\\
	1.66	3.53708508	\\
	1.67	3.543056114	\\
	1.68	3.548992361	\\
	1.69	3.5548942	\\
	1.7	3.560762006	\\
	1.71	3.566596146	\\
	1.72	3.572396987	\\
	1.73	3.578164888	\\
	1.74	3.583900202	\\
	1.75	3.58960328	\\
	1.76	3.595274467	\\
	1.77	3.600914103	\\
	1.78	3.606522523	\\
	1.79	3.612100059	\\
	1.8	3.617647035	\\
	1.81	3.623163775	\\
	1.82	3.628650594	\\
	1.83	3.634107805	\\
	1.84	3.639535717	\\
	1.85	3.644934633	\\
	1.86	3.650304852	\\
	1.87	3.655646669	\\
	1.88	3.660960376	\\
	1.89	3.666246259	\\
	1.9	3.671504602	\\
	1.91	3.676735682	\\
	1.92	3.681939774	\\
	1.93	3.68711715	\\
	1.94	3.692268076	\\
	1.95	3.697392817	\\
	1.96	3.702491631	\\
	1.97	3.707564774	\\
	1.98	3.712612499	\\
	1.99	3.717635055	\\
	2	3.722632686	\\
};
\addlegendentry{$\mathsf{H}_{\mbox{\tiny TH-1}}(\sigma,\theta)$}

\addplot[smooth,color=black!20!blue,thick,dash pattern={on 4pt off 2pt}]
table[row sep=crcr]
{
0	3.031698717	\\
0.1	3.031780135	\\
0.2	3.032583748	\\
0.3	3.035623569	\\
0.4	3.042924633	\\
0.5	3.056437898	\\
0.6	3.077507551	\\
0.7	3.106567308	\\
0.8	3.143139098	\\
0.9	3.186071916	\\
1	3.233877254	\\
1.1	3.285028322	\\
1.2	3.338156584	\\
1.3	3.392142787	\\
1.4	3.446132032	\\
1.5	3.499506965	\\
1.6	3.551844618	\\
1.7	3.602871823	\\
1.8	3.65242612	\\
1.9	3.700424292	\\
2	3.746838328	\\
};
\addlegendentry{$\mathsf{H}_{\mbox{\tiny TH-3}}(\sigma,\theta)$}

\end{axis}
\end{tikzpicture}
		\caption{For $\sigma=5$} 
		\label{fig:POPULATION2}
	\end{subfigure}%
	\caption{Comparison between the bounds $\mathsf{H}_{\mbox{\tiny TH-1}}(\sigma,\theta)$, $\mathsf{H}_{\mbox{\tiny TH-3}}(\sigma,\theta)$, and the trivial bound $H(Y_{n+1}|Y_{n})$ for a quantized MA process with $\theta$ ranging in $[0,2]$ and two values of $\sigma$.} 
	\label{fig:Gauss}
\end{figure}
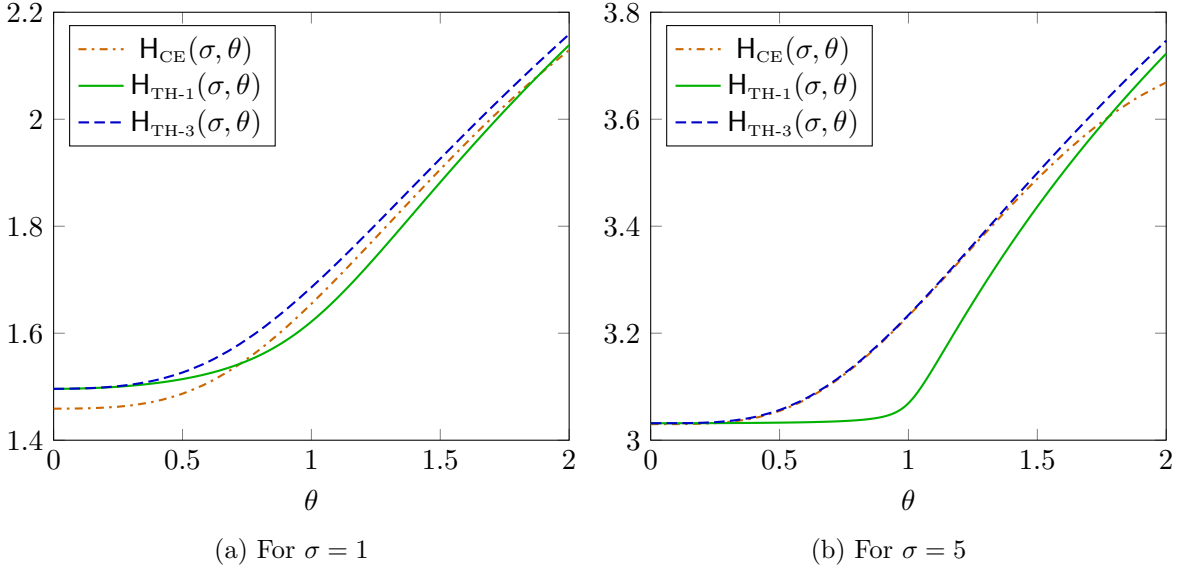

\subsection{The Quantized-Hidden Autoregressive Process} \label{SEC_Example2}

We continue by exhibiting the usefulness of the upper bound in Theorem \ref{Theorem_t_Distribution_Order_K} by referring to a quantized version of a hidden AR process.
Let $\{W_{n}\}$ be a white noise process, defined by $W_{n} \sim \calN(0,\sigma^{2})$, and for $\varphi \in \mathbb{R}$, let $\{X_{n}\}$ be a first-order AR process, defined by
\begin{align}
X_{n} = \varphi X_{n-1} + W_{n}.
\end{align}
This process is Gaussian and stationary for any $\varphi \in (-1,1)$, and its covariance function is given by
\begin{align}
R_{\mbox{\tiny X}}(k)
= \frac{\sigma^{2}}{1-\varphi^{2}} \varphi^{|k|}.
\end{align}
The PSD function is given by
\begin{align}
\Phi_{\mbox{\tiny X}}(\lambda) 
= \sum_{k=-\infty}^{\infty} R_{\mbox{\tiny X}}(k) e^{i \lambda k}
= \frac{\sigma^{2}}{1 + \varphi^{2} - 2 \varphi \cos(\lambda)},
\end{align} 
and the differential entropy rate of the process $\{X_{n}\}$ is given by 
\begin{align}
\bar{H}(X) 
&= \frac{1}{2} \log (2 \pi e) 
+ \frac{1}{4\pi} \int_{0}^{2\pi} \log \Phi_{\mbox{\tiny X}}(\lambda) \dint\lambda \\
%%%%%%%%%%%%%%%%%%%%%%%%%%%%%%%%%%%%%%%%%%%%%%%%%%
&= \frac{1}{2} \log (2 \pi e \sigma^{2}) 
+ \frac{1}{4\pi} \int_{0}^{2\pi} \log (1 + \varphi^{2} - 2 \varphi \cos(\lambda)) \dint\lambda \\
%%%%%%%%%%%%%%%%%%%%%%%%%%%%%%%%%%%%%%%%%%%%%%%%%%
&= \frac{1}{2} \log (2 \pi e \sigma^{2}).
\end{align}

Let $\{V_{n}\}$ be another white noise process, defined by $V_{n} \sim \calN(0,\nu^{2})$, and independent of $\{W_{n}\}$. Define the hidden AR process by 
\begin{align}
U_{n} = X_{n} + V_{n}.
\end{align}
Define the quantization function $\mathsf{Q}(\cdot)$ as in \eqref{DEF_QUAN} and define the quantized-hidden AR process $\{Y_{n}\}$ by $Y_{n} = \mathsf{Q}(U_{n})$ at any time $n$.  
As explained for the quantized MA process, the quantized AR process is strongly stationary and its PSD function is well defined. 
In order to calculate the covariance function of the quantized process, let us consider an AR process whose initial value is given by $X_{0}$, and choose $X_{0} \sim \calN(0,\sigma_{0}^{2})$, with $\sigma_{0}^{2} = \tfrac{\sigma^{2}}{1-\varphi^{2}}$, which ensures stationarity.
For $R_{\mbox{\tiny Y}}(0)$, note that $\Exp[\mathsf{Q}(U_{0})]=0$, and then, according to \eqref{Second_Moment_Quan_Gauss}, 
\begin{align}
R_{\mbox{\tiny Y}}(0)
&=\Exp[\mathsf{Q}(U_{0})^{2}] \\
&= \sum_{k=1}^{\infty} 2k^{2} \left[\Phi\left(\frac{k+\frac{1}{2}}{\sqrt{\sigma_{0}^{2}+\nu^{2}}}\right) - \Phi\left(\frac{k-\frac{1}{2}}{\sqrt{\sigma_{0}^{2}+\nu^{2}}}\right)\right] \\
%%%%%%%%%%%%%%%%%%%%%%%%%%%%%%%%%%%%%%%%%%%%%%%%%%%%
&\dfn \mathsf{R}_{0}(\sigma,\varphi,\nu).
\end{align}
For $R_{\mbox{\tiny Y}}(k)$, $k \geq 1$, first note that
\begin{align}
X_{1} &= \varphi X_{0} + W_{1} \\
X_{2} &= \varphi X_{1} + W_{2} = \varphi^{2} X_{0} + \varphi W_{1} + W_{2},
\end{align}
and for a general $k \geq 1$,
\begin{align} 
X_{k} 
= \varphi X_{k-1} + W_{k} 
=\varphi^{k} X_{0} + \sum_{m=0}^{k-1} \varphi^{m} W_{k-m} 
\dfn \varphi^{k} X_{0} + \tilde{W}_{k},
\end{align}
where $\tilde{W}_{k}$ is a Gaussian random variable with zero mean and a variance of
\begin{align}
\text{Var}(\tilde{W}_{k}) 
&= \sigma^{2} \left(1 + \varphi^{2} + \varphi^{4} + \ldots + \varphi^{2(k-1)} \right) \\
&= \sigma^{2} \frac{1-\varphi^{2k}}{1-\varphi^{2}} \\
&\dfn \sigma_{k}^{2}. 
\end{align}
Then,
\begin{align}
R_{\mbox{\tiny Y}}(k)
&=\Exp[\mathsf{Q}(U_{0}) \mathsf{Q}(U_{k})] \\
%%%%%%%%%%%%%%%%%%%%%%%%%%%%%%%%%%%%%%%%%%%%%%%%%
&=\Exp[\mathsf{Q}(X_{0} + V_{0}) \mathsf{Q}(\varphi^{k} X_{0} + \tilde{W}_{k} + V_{k})] \\
%%%%%%%%%%%%%%%%%%%%%%%%%%%%%%%%%%%%%%%%%%%%%%%%%
&= \int_{-\infty}^{\infty} \Exp[\mathsf{Q}(s + V_{0}) \mathsf{Q}(\varphi^{k} s + \tilde{W}_{k} + V_{k})] \tfrac{1}{\sqrt{2 \pi \sigma_{0}^{2}}} e^{-s^{2}/2\sigma_{0}^{2}} \mathrm{d}s \\
%%%%%%%%%%%%%%%%%%%%%%%%%%%%%%%%%%%%%%%%%%%%%%%%%
\label{ref1}
&= \int_{-\infty}^{\infty} \Exp[\mathsf{Q}(s + V_{0})] \Exp[\mathsf{Q}(\varphi^{k}s + \tilde{W}_{k} + V_{k})] \tfrac{1}{\sqrt{2 \pi \sigma_{0}^{2}}} e^{-s^{2}/2\sigma_{0}^{2}} \mathrm{d}s.
\end{align}
It follows from \eqref{ref2} that
\begin{align}
\Exp[\mathsf{Q}(s + V_{0})]
&= \sum_{\ell=-\infty}^{\infty} \ell \left[\Phi\left(\frac{\ell - s +\frac{1}{2}}{\nu}\right) - \Phi\left(\frac{\ell - s - \frac{1}{2}}{\nu}\right)\right] \\
\label{Tosubs1}
&\dfn \sum_{\ell=-\infty}^{\infty} \ell S_{0}(\ell,s,\nu),
\end{align} 
and 
\begin{align}
\Exp[\mathsf{Q}(\varphi^{k}s + \tilde{W}_{k} + V_{k})]
&= \sum_{m=-\infty}^{\infty} m \left[\Phi\left(\frac{m - \varphi^{k}s +\frac{1}{2}}{\sqrt{\sigma_{k}^{2}+\nu^{2}}}\right) - \Phi\left(\frac{m - \varphi^{k}s - \frac{1}{2}}{\sqrt{\sigma_{k}^{2}+\nu^{2}}}\right)\right] \\
\label{Tosubs2}
&\dfn \sum_{m=-\infty}^{\infty} m T_{k}(m,s,\sigma,\varphi,\nu).
\end{align}
By substituting \eqref{Tosubs1} and \eqref{Tosubs2} back into \eqref{ref1}, we arrive at
\begin{align}
R_{\mbox{\tiny Y}}(k)
&= \int_{-\infty}^{\infty} \left(\sum_{\ell=-\infty}^{\infty} \ell S_{0}(\ell,s,\nu)\right) \left(\sum_{m=-\infty}^{\infty} m T_{k}(m,s,\sigma,\varphi,\nu) \right) \tfrac{1}{\sqrt{2 \pi \sigma_{0}^{2}}} e^{-s^{2}/2\sigma_{0}^{2}} \mathrm{d}s \\
%%%%%%%%%%%%%%%%%%%%%%%%%%%%%%%%%%%%%%%%%%%%%%%%%
\label{AR_Covariance}
&= \sum_{\ell=-\infty}^{\infty} \sum_{m=-\infty}^{\infty} \ell m \int_{-\infty}^{\infty} S_{0}(\ell,s,\nu) T_{k}(m,s,\sigma,\varphi,\nu) \tfrac{1}{\sqrt{2 \pi \sigma_{0}^{2}}} e^{-s^{2}/2\sigma_{0}^{2}} \mathrm{d}s \\
%%%%%%%%%%%%%%%%%%%%%%%%%%%%%%%%%%%%%%%%%%%%%%%%%%%%
&\dfn \mathsf{R}_{k}(\sigma,\varphi,\nu).
\end{align}
It follows from Theorem \ref{Theorem_t_Distribution_Order_K} that the entropy rate of the quantized-hidden AR process is upper-bounded by
\begin{align}  
\inf_{\{\bbeta \in \mathbb{R}^{k}:~\sum_{m=1}^{k} |\beta_{m}| < 1 \}} 
\frac{1}{4\pi} \int_{0}^{2\pi} \log \left[2\pi e \frac{\left(\mathsf{R}_{0}(\sigma,\varphi,\nu) + \tfrac{1}{12}\right) + \sum_{m=1}^{k} \beta_{m} \mathsf{R}_{m}(\sigma,\varphi,\nu)}{1 + \sum_{m=1}^{k} \beta_{m} \cos (m \lambda)} \right]  \dint\lambda.
\end{align}   
Let us denote this bound by the function $\mathsf{H}_{\mbox{\tiny TH-2}}^{\mbox{\tiny AR}}(\sigma,\varphi,\nu,k)$.
Regarding the trivial bound $H(Y_{1}|Y_{0})$, first note that the marginal distribution of $Y_{0}$ is  
\begin{align}
P_{Y_{0}}(k) = \Phi\left(\frac{k+\frac{1}{2}}{\sqrt{\sigma_{0}^{2}+\nu^{2}}}\right) - \Phi\left(\frac{k-\frac{1}{2}}{\sqrt{\sigma_{0}^{2}+\nu^{2}}}\right).
\end{align}
The joint distribution of $(Y_{0},Y_{1})$ is given by
\begin{align}
&P_{Y_{0}Y_{1}}(k,\ell) \nn \\
&= \prob \left\{Y_{0}=k,~Y_{1}=\ell \right\} \\
%%%%%%%%%%%%%%%%%%%%%%%%%%%%%%%%%%%%%%%%%%%%%%%%%%
&= \prob \left\{k-\tfrac{1}{2} \leq U_{0} \leq k+\tfrac{1}{2},~\ell-\tfrac{1}{2} \leq U_{1} \leq \ell+\tfrac{1}{2} \right\} \\
%%%%%%%%%%%%%%%%%%%%%%%%%%%%%%%%%%%%%%%%%%%%%%%%%%
&= \prob \left\{k-\tfrac{1}{2} \leq X_{0} + V_{0} \leq k+\tfrac{1}{2},~\ell-\tfrac{1}{2} \leq \varphi X_{0} + W_{1} + V_{1} \leq \ell+\tfrac{1}{2} \right\} \\
%%%%%%%%%%%%%%%%%%%%%%%%%%%%%%%%%%%%%%%%%%%%%%%%%%
&= \int_{-\infty}^{\infty} \prob \left\{k-\tfrac{1}{2} \leq s + V_{0} \leq k+\tfrac{1}{2},~\ell-\tfrac{1}{2} \leq \varphi s + W_{1} + V_{1} \leq \ell+\tfrac{1}{2} \right\} \tfrac{1}{\sqrt{2 \pi \sigma_{0}^{2}}} e^{-s^{2}/2\sigma_{0}^{2}} \mathrm{d}s \\
%%%%%%%%%%%%%%%%%%%%%%%%%%%%%%%%%%%%%%%%%%%%%%%%%%
&= \int_{-\infty}^{\infty} \prob \left\{k-\tfrac{1}{2} \leq s + V_{0} \leq k+\tfrac{1}{2} \right\} \prob \left\{\ell-\tfrac{1}{2} \leq \varphi s + W_{1} + V_{1} \leq \ell+\tfrac{1}{2} \right\} \tfrac{1}{\sqrt{2 \pi \sigma_{0}^{2}}} e^{-s^{2}/2\sigma_{0}^{2}} \mathrm{d}s \\
%%%%%%%%%%%%%%%%%%%%%%%%%%%%%%%%%%%%%%%%%%%%%%%%%% 
&= \int_{-\infty}^{\infty} S_{0}(k,s,\nu) T_{1}(\ell,s,\sigma,\varphi,\nu) \tfrac{1}{\sqrt{2 \pi \sigma_{0}^{2}}} e^{-s^{2}/2\sigma_{0}^{2}} \mathrm{d}s.
\end{align}
Now, the conditional entropy is given by
\begin{align}
H(Y_{1}|Y_{0}) = - \sum_{k=-\infty}^{\infty} \sum_{\ell=-\infty}^{\infty} P_{Y_{0}Y_{1}}(k,\ell) \log \frac{P_{Y_{0}Y_{1}}(k,\ell)}{P_{Y_{0}}(k)},
\end{align}
which is merely a function of $\sigma$, $\varphi$, and $\nu$, to be denoted by $\mathsf{H}_{\mbox{\tiny CE}}^{\mbox{\tiny AR}}(\sigma,\varphi,\nu)$.

We now compare numerically the bound $\mathsf{H}_{\mbox{\tiny TH-2}}^{\mbox{\tiny AR}}(\sigma,\varphi,\nu,k)$ with $k=2,3$ and the trivial bound $H(Y_{1}|Y_{0})$ in the specific case of $\sigma=1$ and $\nu=4$.
As can be seen in Figure \ref{fig:QH_AR}, in the range of relatively high $\varphi$ values, the new upper bounds from Theorem \ref{Theorem_t_Distribution_Order_K} outperform the conditional entropy upper bound. Still, for relatively low values of $\varphi$, the conditional entropy bound is lower than the bounds from Theorem \ref{Theorem_t_Distribution_Order_K}, which is not very surprising, since at the extreme of $\varphi=0$, the process $\{Y_{n}\}$ is i.i.d.\ and $H(Y_{1}|Y_{0})$ yields an exact estimation for the entropy rate.
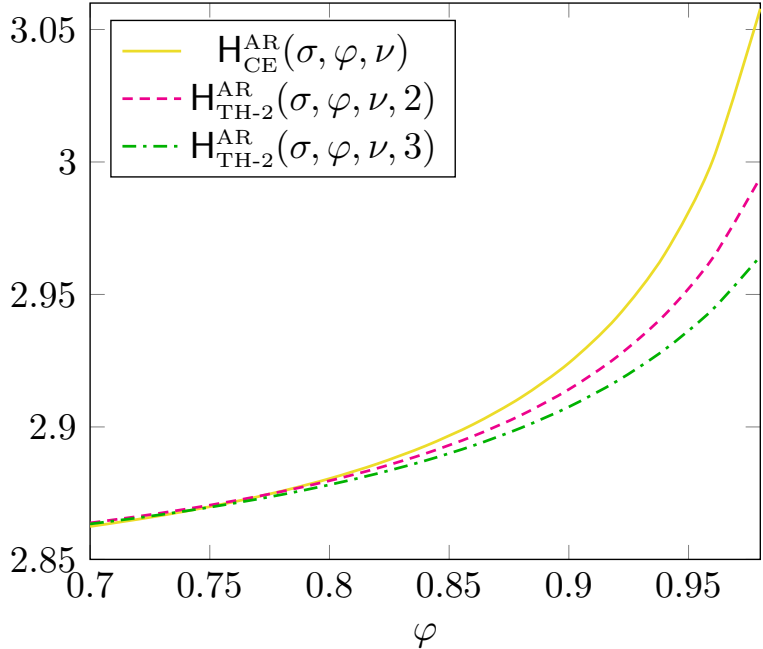
\begin{figure}[ht!]
	\centering
	\begin{tikzpicture}[scale=1.3]
	\begin{axis}[
	disabledatascaling,
	%x=25cm,
	%y=30cm,
	scaled x ticks=false,
	xticklabel style={/pgf/number format/fixed,
		/pgf/number format/precision=3},
	scaled y ticks=false,
	yticklabel style={/pgf/number format/fixed,
		/pgf/number format/precision=3},
	xlabel={$\varphi$},
	xmin=0.7, xmax=0.98,
	ymin=2.85, ymax=3.06,
	legend pos=north west,
	%    ymajorgrids=true,
	%    grid style=dashed,
	]
	
	\addplot[smooth,color=black!10!yellow,thick]
	table[row sep=crcr] 
	{
0.7	2.862446961	\\
0.72	2.865141897	\\
0.74	2.868205217	\\
0.76	2.871713821	\\
0.78	2.875767852	\\
0.8	2.880500188	\\
0.82	2.886091013	\\
0.84	2.892790964	\\
0.86	2.900959357	\\
0.88	2.911130436	\\
0.9	2.924135013	\\
0.92	2.941340606	\\
0.94	2.965172027	\\
0.96	3.000391919	\\
0.98	3.057882965	\\
	};
	\legend{}
	\addlegendentry{$\mathsf{H}_{\mbox{\tiny CE}}^{\mbox{\tiny AR}}(\sigma,\varphi,\nu)$}	
	
	\addplot[smooth,color=magenta,thick,dash pattern={on 3pt off 2pt}]
	table[row sep=crcr]
	{
0.7	2.863790438	\\
0.72	2.866232262	\\
0.74	2.868980731	\\
0.76	2.872089599	\\
0.78	2.875625822	\\
0.8	2.879695495	\\
0.82	2.884414328	\\
0.84	2.889945206	\\
0.86	2.89651468	\\
0.88	2.904432844	\\
0.9	2.914168837	\\
0.92	2.926411656	\\
0.94	2.94226806	\\
0.96	2.963638167	\\
0.98	2.994138702	\\
	};
	\addlegendentry{$\mathsf{H}_{\mbox{\tiny TH-2}}^{\mbox{\tiny AR}}(\sigma,\varphi,\nu,2)$}
	
	\addplot[smooth,color=black!30!green,thick,dash pattern={on 4pt off 2pt on 1pt off 2pt}]
	table[row sep=crcr]
	{
		0.7	2.863388083	\\
		0.72	2.865702693	\\
		0.74	2.86828803	\\
		0.76	2.871194459	\\
		0.78	2.874454786	\\
		0.8	2.878167076	\\
		0.82	2.882409489	\\
		0.84	2.887272278	\\
		0.86	2.892948241	\\
		0.88	2.899646902	\\
		0.9	2.907583792	\\
		0.92	2.91724992	\\
		0.94	2.92917668	\\
		0.96	2.944288852	\\
		0.98	2.964432509	\\
	};
	\addlegendentry{$\mathsf{H}_{\mbox{\tiny TH-2}}^{\mbox{\tiny AR}}(\sigma,\varphi,\nu,3)$}
	
	\end{axis}
	\end{tikzpicture}
	\caption{Comparison between the bounds $\mathsf{H}_{\mbox{\tiny TH-2}}^{\mbox{\tiny AR}}(\sigma,\varphi,\nu,k)$, $k=2,3$, and the trivial bound $H(Y_{1}|Y_{0})$ for a quantized-hidden AR process with $\sigma=1, \nu=4$, and $\varphi$ ranging in $[0.7,0.98]$.}\label{fig:QH_AR}
\end{figure}

\section*{Appendix A - Proof of Theorem \ref{Theorem_Gaussian}}
\renewcommand{\theequation}{A.\arabic{equation}}
\setcounter{equation}{0}

We have the following 
\begin{align}
H(\bY_{n}) 
&= h(\tilde{\bY}_{n}) \\
&\leq \frac{1}{2} \log \left[(2 \pi e)^{n} \det\left(K_{\tilde{\bY}_{n}}\right) \right] \\
&= \frac{1}{2} \log \left[(2 \pi e)^{n} \det\left(K_{\bY_{n}} + K_{\bU_{n}}\right) \right] \\
\label{refb1}
&= \frac{1}{2} \log \left[(2 \pi e)^{n}  \det\left(K_{\bY_{n}} + \frac{1}{12}I_{n}\right) \right],
\end{align}
where $I_{n}$ denotes the $n \times n$ identity matrix. 
In order to deal with the determinant of a sum of two matrices, we invoke the following result. 

\begin{proposition}[\cite{Fiedler}]
	Let $A$ and $B$ be hermitian $n \times n$ matrices with eigenvalues $\alpha_{1} \geq \alpha_{2} \geq \cdots \alpha_{n}$ and $\beta_{1} \geq \beta_{2} \geq \cdots \beta_{n}$ respectively. Then
	\begin{align}
	\min_{\pi} \prod_{i=1}^{n} (\alpha_{i} + \beta_{\pi_{i}})
	\leq \det(A+B)
	\leq \max_{\pi} \prod_{i=1}^{n} (\alpha_{i} + \beta_{\pi_{i}}),
	\end{align} 
	where the minimum or maximum is taken over all permutations of indices $1,2,\ldots,n$.
	
	In particular, if $\alpha_{n}+\beta_{n} \geq 0$, which is certainly true if both $A$ and $B$ are positive semidefinite, then 
	\begin{align} \label{refb2}
	\prod_{i=1}^{n} (\alpha_{i} + \beta_{i})
	\leq \det(A+B)
	\leq \prod_{i=1}^{n} (\alpha_{i} + \beta_{n+1-i}).
	\end{align}
	This estimates are best possible in terms of the eigenvalues of $A$ and $B$.
\end{proposition}

Since both matrices in \eqref{refb1} are positive semidefinite (a covariance matrix is always positive semidefinite), we may use the upper bound in \eqref{refb2}. Let $\{\tau_{n,i};~1 \leq i \leq n\}$ be the eigenvalues of the covariance matrix $K_{\bY_{n}}$, then
\begin{align} \label{refb3}
\det\left(K_{\bY_{n}} + \frac{1}{12}I_{n}\right)
\leq \prod_{i=1}^{n} \left(\tau_{n,i} + \frac{1}{12}\right).
\end{align}   
Upper-bounding \eqref{refb1} by \eqref{refb3} yields  
\begin{align}
H(\bY_{n}) 
&\leq \frac{1}{2} \log \left[(2 \pi e)^{n}  \prod_{i=1}^{n} \left(\tau_{n,i} + \frac{1}{12}\right) \right] \\
&= \frac{n}{2} \log (2 \pi e) + \frac{1}{2} \sum_{i=1}^{n} \log \left(\tau_{n,i} + \frac{1}{12}\right). 
\end{align}
Taking now the limit, we find that the entropy rate is upper-bounded by
\begin{align}
\bar{H}(Y) 
&= \lim_{n \to \infty} \frac{1}{n} H(\bY_{n}) \\
&\leq \lim_{n \to \infty} \frac{1}{n} \left[\frac{n}{2} \log (2 \pi e) + \frac{1}{2} \sum_{i=1}^{n} \log \left(\tau_{n,i} + \frac{1}{12}\right)\right] \\
&= \frac{1}{2} \log (2 \pi e) 
+ \lim_{n \to \infty} \frac{1}{2n} \sum_{i=1}^{n} \log \left(\tau_{n,i} + \frac{1}{12}\right).
\end{align}
Finally, since $K_{\bY_{n}}$ is a Toeplitz matrix, it follows from Szeg\"o theorem \cite{Szego} that
\begin{align}
\bar{H}(Y) 
&\leq \frac{1}{2} \log (2 \pi e) 
+ \frac{1}{4\pi} \int_{0}^{2\pi} \log \left(\Phi_{\mbox{\tiny Y}}(\lambda) + \frac{1}{12}\right) \dint\lambda,
\end{align}
where $\Phi_{\mbox{\tiny Y}}(\lambda)$ is the power spectral density function of the process $\{Y_{n}\}_{n \geq 1}$.

\section*{Appendix B - Proof of Theorem \ref{Theorem_t_Distribution_Order_K}}
\renewcommand{\theequation}{B.\arabic{equation}}
\setcounter{equation}{0}

Trivially,
\begin{align}
- \log q_{\mbox{\tiny t}}(\by)
= -\log C_{n} +\tfrac{1}{2}\log |\bSigma|
+ \tfrac{1}{2}(\nu+n) \log \left[1+\tfrac{1}{\nu}(\by-\bmu)^{\mbox{\tiny T}}\bSigma^{-1}(\by-\bmu)\right], 
\end{align} 
and so
\begin{align}
&- \Exp\left[\log q_{\mbox{\tiny t}}(\tilde{\bY}_{n})\right] \nn \\
&= -\log C_{n} +\tfrac{1}{2}\log |\bSigma|
+ \tfrac{1}{2}(\nu+n) \Exp \left\{\log \left[1+\tfrac{1}{\nu}(\tilde{\bY}_{n}-\bmu)^{\mbox{\tiny T}}\bSigma^{-1}(\tilde{\bY}_{n}-\bmu)\right] \right\} \\
%%%%%%%%%%%%%%%%%%%%%%%%%%%%%%%%%%%%%%%%%%%%%%%%%
\label{ref4}
&\leq -\log C_{n} +\tfrac{1}{2}\log |\bSigma|
+ \tfrac{1}{2}(\nu+n) \log \left[1+\tfrac{1}{\nu} \Exp \left\{ (\tilde{\bY}_{n}-\bmu)^{\mbox{\tiny T}}\bSigma^{-1}(\tilde{\bY}_{n}-\bmu)\right\} \right],
\end{align}
where \eqref{ref4} follows from Jensen's inequality and the concavity of the logarithmic function. 

In order to confine ourself to a relatively small set of free parameters (to be optimized eventually), let us consider a symmetric banded Toeplitz matrix $\bSigma^{-1} = \{a_{ij}\}$, $i,j \in \{1,2,\ldots,n\}$, such that $a_{i,j} = \alpha_{|i-j|}$ as long as $|i-j| \leq k$, and zero otherwise. We assume that $\alpha_{0}>0$ and denote $\balpha = (\alpha_{0},\alpha_{1},\ldots,\alpha_{k})$.  
Since every positive definite matrix is invertible and its inverse is also positive definite \cite[p.\ 438, Theorem 7.2.1]{MATRIX_Analysis}, it is enough to require that $\bSigma^{-1}$ is positive definite.
We have the following result concerning symmetric Toeplitz matrices, which follows directly from \cite[p.\ 349, Theorem 6.1.10]{MATRIX_Analysis}. 

\begin{proposition} \label{Toeplitz_Lemma}
	If $M$ is a symmetric Toeplitz matrix, i.e., the entries $m_{ij}$ are given as a function of their absolute index differences: $m_{ij} = h(|i-j|)$, and the strict inequality 
	\begin{align}
	\sum_{j \neq 0} |h(j)| < h(0)
	\end{align}
	holds, then $M$ is strictly positive definite.
\end{proposition}

From the sufficient conditions in Proposition \ref{Toeplitz_Lemma}, we conclude that $\bSigma^{-1}$ is positive definite as long as $\sum_{j=1}^{k} |\alpha_{j}| < \alpha_{0}/2$, or equivalently, $\sum_{j=1}^{k} \left|\frac{2\alpha_{j}}{\alpha_{0}}\right| < 1$.  

Due to the fact that the components of $\bU_{n}$ are uniformly distributed in $[0,1)$, let us choose $\bmu = [\mu_{\mbox{\tiny Y}}+\tfrac{1}{2}~\mu_{\mbox{\tiny Y}}+\tfrac{1}{2}~\ldots~\mu_{\mbox{\tiny Y}}+\tfrac{1}{2}]^{\mbox{\tiny T}}$, where $\mu_{\mbox{\tiny Y}}$ is the expectation of the process $\{Y_{n}\}$, and then
\begin{align}
&\Exp \left\{ (\tilde{\bY}_{n}-\bmu)^{\mbox{\tiny T}}\bSigma^{-1}(\tilde{\bY}_{n}-\bmu)\right\} \nn \\
&= \Exp \left\{ \alpha_{0} \sum_{i=1}^{n} (\tilde{Y}_{i}-\mu_{\mbox{\tiny Y}}-\tfrac{1}{2})^{2} + \sum_{j=1}^{k} 2 \alpha_{j} \sum_{i=1}^{n-j} (\tilde{Y}_{i}-\mu_{\mbox{\tiny Y}}-\tfrac{1}{2}) (\tilde{Y}_{i+j}-\mu_{\mbox{\tiny Y}}-\tfrac{1}{2}) \right\} \\
\label{ref7}
&=  
\alpha_{0} \sum_{i=1}^{n} \Exp \left\{(\tilde{Y}_{i}-\mu_{\mbox{\tiny Y}}-\tfrac{1}{2})^{2} \right\} 
+ \sum_{j=1}^{k} 2 \alpha_{j} \sum_{i=1}^{n-j} \Exp \left\{ (\tilde{Y}_{i}-\mu_{\mbox{\tiny Y}}-\tfrac{1}{2}) (\tilde{Y}_{i+j}-\mu_{\mbox{\tiny Y}}-\tfrac{1}{2}) \right\}. 
\end{align}
For the first term in \eqref{ref7},
\begin{align}
\Exp \left\{(\tilde{Y}_{i}-\mu_{\mbox{\tiny Y}}-\tfrac{1}{2})^{2} \right\}
&= \Exp \left\{(Y_{i}-\mu_{\mbox{\tiny Y}}+U_{i}-\tfrac{1}{2})^{2} \right\}\\
&= \Exp \left\{(Y_{i}-\mu_{\mbox{\tiny Y}})^{2} + 2(Y_{i}-\mu_{\mbox{\tiny Y}})(U_{i}-\tfrac{1}{2}) +    (U_{i}-\tfrac{1}{2})^{2} \right\}\\
&= \Exp \left\{(Y_{i}-\mu_{\mbox{\tiny Y}})^{2}\right\} + \Exp \left\{  (U_{i}-\tfrac{1}{2})^{2} \right\}\\
\label{help4}
&= R_{\mbox{\tiny Y}}(0) + \tfrac{1}{12},
\end{align}
and for the second term in \eqref{ref7},
\begin{align}
&\Exp \left\{ (\tilde{Y}_{i}-\mu_{\mbox{\tiny Y}}-\tfrac{1}{2}) (\tilde{Y}_{i+j}-\mu_{\mbox{\tiny Y}}-\tfrac{1}{2}) \right\} \nn \\
&= \Exp \left\{ (Y_{i}-\mu_{\mbox{\tiny Y}}+U_{i}-\tfrac{1}{2}) (Y_{i+j}-\mu_{\mbox{\tiny Y}}+U_{i+j}-\tfrac{1}{2}) \right\} \\
&= \Exp \left\{ (Y_{i}-\mu_{\mbox{\tiny Y}})(Y_{i+j}-\mu_{\mbox{\tiny Y}}) + (Y_{i}-\mu_{\mbox{\tiny Y}})(U_{i+j}-\tfrac{1}{2}) \right. \nn \\
&\left.~~~~~~~~ + (Y_{i+j}-\mu_{\mbox{\tiny Y}})(U_{i}-\tfrac{1}{2}) + (U_{i}-\tfrac{1}{2})(U_{i+j}-\tfrac{1}{2}) \right\} \\
&= \Exp \left\{ (Y_{i}-\mu_{\mbox{\tiny Y}})(Y_{i+j}-\mu_{\mbox{\tiny Y}}) \right\} \\
\label{help5}
&= R_{\mbox{\tiny Y}}(j).
\end{align}
Substituting it back into \eqref{ref7} yields
\begin{align}
\Exp \left\{ (\tilde{\bY}_{n}-\bmu)^{\mbox{\tiny T}}\bSigma^{-1}(\tilde{\bY}_{n}-\bmu)\right\}
&= \alpha_{0} n \left(R_{\mbox{\tiny Y}}(0) + \tfrac{1}{12}\right) + \sum_{j=1}^{k}2\alpha_{j} (n-j) R_{\mbox{\tiny Y}}(j) \\
&\leq \alpha_{0} n \left(R_{\mbox{\tiny Y}}(0) + \tfrac{1}{12}\right) + 2n \sum_{j=1}^{k}\alpha_{j} R_{\mbox{\tiny Y}}(j) \\
&= n \left[\alpha_{0} \left(R_{\mbox{\tiny Y}}(0) + \tfrac{1}{12}\right) + 2\sum_{j=1}^{k}\alpha_{j} R_{\mbox{\tiny Y}}(j) \right] \\
\label{help1}
&= n f(\balpha).
\end{align}
Since $\bSigma^{-1}$ is positive definite, its determinant is nonzero, and the determinant of $\bSigma$ is given by
\begin{align}
\det(\bSigma) = [\det(\bSigma^{-1})]^{-1}.
\end{align} 
At the moment, let us express the determinant of $\bSigma^{-1}$ as a product of the eigenvalues $\{\tau_{n,k}\}_{k=1}^{n}$ of $\bSigma^{-1}$, such that the second term in \eqref{ref4} is given by
\begin{align}
\frac{1}{2}\log |\bSigma|
&= -\frac{1}{2}\log |\bSigma^{-1}| \\
&= -\frac{1}{2}\log \left(\prod_{k=1}^{n} \tau_{n,k} \right) \\
\label{help2}
&= -\frac{1}{2} \sum_{k=1}^{n} \log \left( \tau_{n,k} \right).
\end{align}   
We now upper-bound \eqref{ref4} using \eqref{help1} and \eqref{help2}, and also substitute the expression for the normalizing constant from \eqref{Normalizing}. We divide by $n$ and arrive at    
\begin{align}
&-\frac{1}{n} \Exp\left[\log q_{\mbox{\tiny t}}(\tilde{\bY}_{n})\right] \nn \\
%%%%%%%%%%%%%%%%%%%%%%%%%%%%%%%%%%%%%%%%%%%%%%%%%
%%%%%%%%%%%%%%%%%%%%%%%%%%%%%%%%%%%%%%%%%%%%%%%%%
&\leq - \frac{1}{n} \log \left(\frac{\Gamma[(\nu+n)/2]}{\Gamma(\nu/2)\nu^{n/2}\pi^{n/2}}\right) -\frac{1}{2n} \sum_{k=1}^{n} \log \left( \tau_{n,k} \right)
+ \frac{1}{2}\frac{\nu+n}{n} \log \left(1+\frac{n}{\nu} f(\balpha) \right) \\
%%%%%%%%%%%%%%%%%%%%%%%%%%%%%%%%%%%%%%%%%%%%%%%%%
%%%%%%%%%%%%%%%%%%%%%%%%%%%%%%%%%%%%%%%%%%%%%%%%%
\label{ref10}
&= - \frac{1}{n} \log\Gamma[(\nu+n)/2]
+ \frac{1}{n} \log\Gamma(\nu/2)
+ \frac{1}{2} \log(\nu \pi) \nn \\
&~~~~~~
-\frac{1}{2n} \sum_{k=1}^{n} \log \left( \tau_{n,k} \right)
+ \frac{1}{2}\frac{\nu+n}{n} \log \left(1+\frac{n}{\nu} f(\balpha) \right).
\end{align}
In order to proceed, we invoke the following inequalities from \cite[Lemma 1]{Minc}
\begin{align}
\label{Log_Gamma_Inequalities}
0< \log \Gamma(x) - \left[\left(x-\frac{1}{2}\right)\log(x) - x + \frac{1}{2} \log (2\pi)\right] < \frac{1}{x}, 
\end{align}
which are valid for $x>1$. From the left-hand-side inequality of \eqref{Log_Gamma_Inequalities} we get
\begin{align}
\log \Gamma[(\nu+n)/2] > \left(\frac{\nu+n}{2}-\frac{1}{2}\right)\log\left(\frac{\nu+n}{2}\right) - \frac{\nu+n}{2} + \frac{1}{2} \log (2\pi),
\end{align}
and so,
\begin{align}
\label{help6}
-\frac{1}{n}\log \Gamma[(\nu+n)/2] < \left(\frac{1}{2n}-\frac{\nu+n}{2n}\right)\log\left(\frac{\nu+n}{2}\right) + \frac{\nu+n}{2n} - \frac{1}{2n} \log (2\pi).
\end{align}
Continuing from \eqref{ref10}, 
\begin{align}
&-\frac{1}{n} \Exp\left[\log q_{\mbox{\tiny t}}(\tilde{\bY}_{n})\right] \nn \\
%%%%%%%%%%%%%%%%%%%%%%%%%%%%%%%%%%%%%%%%%%%%%%%%%
%%%%%%%%%%%%%%%%%%%%%%%%%%%%%%%%%%%%%%%%%%%%%%%%%
&\leq \left(\frac{1}{2n}-\frac{\nu+n}{2n}\right)\log\left(\frac{\nu+n}{2}\right) + \frac{\nu+n}{2n} - \frac{1}{2n} \log (2\pi)
+ \frac{1}{n} \log\Gamma(\nu/2)
+ \frac{1}{2} \log(\nu \pi) \nn \\
&~~~~~~
-\frac{1}{2n} \sum_{k=1}^{n} \log \left( \tau_{n,k} \right)
+ \frac{1}{2}\frac{\nu+n}{n} \log \left(1+\frac{n}{\nu} f(\balpha) \right) \\
%%%%%%%%%%%%%%%%%%%%%%%%%%%%%%%%%%%%%%%%%%%%%%%%%
%%%%%%%%%%%%%%%%%%%%%%%%%%%%%%%%%%%%%%%%%%%%%%%%%
&= \frac{1}{2n} \log\left(\frac{\nu+n}{2}\right) + \frac{\nu+n}{2n} - \frac{1}{2n} \log (2\pi)
+ \frac{1}{n} \log\Gamma(\nu/2)
+ \frac{1}{2} \log(\nu \pi) \nn \\
&~~~~~~
-\frac{1}{2n} \sum_{k=1}^{n} \log \left( \tau_{n,k} \right)
+ \frac{1}{2}\frac{\nu+n}{n} \log \left(\frac{2+\frac{2n}{\nu} f(\balpha)}{\nu+n} \right).
\end{align}
Taking now the limit, we find that the entropy rate is upper-bounded by
\begin{align}
\bar{H}(Y) 
&= \lim_{n \to \infty} \frac{1}{n} H(\bY_{n}) \\
&\leq \lim_{n \to \infty} \left\{ \frac{1}{2n} \log\left(\frac{\nu+n}{2}\right) + \frac{\nu+n}{2n} - \frac{1}{2n} \log (2\pi)
+ \frac{1}{n} \log\Gamma(\nu/2)
+ \frac{1}{2} \log(\nu \pi) \right. \nn \\
&~~~~~~~~~~\left.
-\frac{1}{2n} \sum_{k=1}^{n} \log \left( \tau_{n,k} \right)
+ \frac{1}{2}\frac{\nu+n}{n} \log \left(\frac{2+\frac{2n}{\nu} f(\balpha)}{\nu+n} \right) \right\} \\
%%%%%%%%%%%%%%%%%%%%%%%%%%%%%%%%%%%%%%%%%%%%%%%%%%%%
%%%%%%%%%%%%%%%%%%%%%%%%%%%%%%%%%%%%%%%%%%%%%%%%%%%%
\label{ref5}
&= 
\frac{1}{2} 
+ \frac{1}{2} \log(\nu \pi) 
-\frac{1}{4\pi} \int_{0}^{2\pi} \log \Theta(\balpha,\lambda) \dint\lambda
+ \frac{1}{2}\log \left(\frac{2}{\nu} f(\balpha) \right) \\
%%%%%%%%%%%%%%%%%%%%%%%%%%%%%%%%%%%%%%%%%%%%%%%%%%%%
%%%%%%%%%%%%%%%%%%%%%%%%%%%%%%%%%%%%%%%%%%%%%%%%%%%%
&= \frac{1}{2} \log \left(2\pi e f(\balpha)\right)
-\frac{1}{4\pi} \int_{0}^{2\pi} \log \Theta(\balpha,\lambda) \dint\lambda,
\end{align}
where \eqref{ref5} follows from Szeg\"o theorem \cite{Szego} with
\begin{align}
\Theta(\balpha,\lambda)
&= \sum_{m=-k}^{k} \alpha_{m} e^{im \lambda} \\
&= \alpha_{0} + \sum_{m=1}^{k} \alpha_{m} \left(e^{im \lambda} + e^{-im \lambda}\right) \\ 
&= \alpha_{0} + \sum_{m=1}^{k} 2 \alpha_{m} \cos (m \lambda). 
\end{align}
As a final step, we may optimize over $\balpha$ to get the tightest possible bound:
\begin{align} \label{ToRef1}
\bar{H}(Y) 
\leq \inf_{\{\balpha\in \mathbb{R}^{k+1}:~ \alpha_{0} > 0,~\sum_{j=1}^{k} \left|2\alpha_{j}/\alpha_{0}\right| < 1 \}} \left\{ \frac{1}{2} \log \left(2\pi e f(\balpha)\right)
-\frac{1}{4\pi} \int_{0}^{2\pi} \log \Theta(\balpha,\lambda) \dint\lambda \right\}.
\end{align}
This minimization problem can be simplified as follows. Note that the objective function in \eqref{ToRef1} can also be written as
\begin{align}
\frac{1}{2} \log \left(2\pi e f(\balpha)\right)
-\frac{1}{4\pi} \int_{0}^{2\pi} \log \Theta(\balpha,\lambda) \dint\lambda
= \frac{1}{4\pi} \int_{0}^{2\pi} \log \left[\frac{2\pi e f(\balpha)}{\Theta(\balpha,\lambda)}\right]  \dint\lambda.
\end{align} 
Now,
\begin{align}
\frac{f(\balpha)}{\Theta(\balpha,\lambda)}
%%%%%%%%%%%%%%%%%%%%%%%%%%%%%%%%%%%%%%%%%%%%%%%%%%%%
&= \frac{\alpha_{0} \left(R_{\mbox{\tiny Y}}(0) + \tfrac{1}{12}\right) + \sum_{m=1}^{k} 2\alpha_{m} R_{\mbox{\tiny Y}}(m)}{\alpha_{0} + \sum_{m=1}^{k} 2 \alpha_{m} \cos (m \lambda)} \\
%%%%%%%%%%%%%%%%%%%%%%%%%%%%%%%%%%%%%%%%%%%%%%%%%%%%
&= \frac{\alpha_{0}\left[\left(R_{\mbox{\tiny Y}}(0) + \tfrac{1}{12}\right) + \sum_{m=1}^{k} \frac{2\alpha_{m}}{\alpha_{0}} R_{\mbox{\tiny Y}}(m)\right]}{\alpha_{0}\left[1 + \sum_{m=1}^{k} \frac{2 \alpha_{m}}{\alpha_{0}} \cos (m \lambda)\right]} \\
%%%%%%%%%%%%%%%%%%%%%%%%%%%%%%%%%%%%%%%%%%%%%%%%%%%%
\label{ToRef2}
&= \frac{\left(R_{\mbox{\tiny Y}}(0) + \tfrac{1}{12}\right) + \sum_{m=1}^{k} \frac{2\alpha_{m}}{\alpha_{0}} R_{\mbox{\tiny Y}}(m)}{1 + \sum_{m=1}^{k} \frac{2 \alpha_{m}}{\alpha_{0}} \cos (m \lambda)}.
\end{align}
Let us denote $\beta_{m} = \frac{2\alpha_{m}}{\alpha_{0}}$, $m=1,2,\ldots,k$.
Since the expression in \eqref{ToRef2} depends on the parameters $\{\alpha_{0}, \alpha_{1}, \ldots, \alpha_{k}\}$ only via $\bbeta = \{\beta_{1}, \ldots , \beta_{k}\}$, then the minimization problem in \eqref{ToRef1} is equivalent to 
\begin{align} 
\bar{H}(Y) 
\leq \inf_{\{\bbeta \in \mathbb{R}^{k}:~\sum_{m=1}^{k} |\beta_{m}| < 1 \}} \left\{ \frac{1}{2} \log \left(2\pi e \Sigma(\bbeta)\right)
-\frac{1}{4\pi} \int_{0}^{2\pi} \log \Psi(\bbeta,\lambda) \dint\lambda \right\},
\end{align}
where $\Sigma(\bbeta)$ and $\Psi(\bbeta,\lambda)$ are defined in \eqref{Def_Sigma} and \eqref{Def_Psi}, respectively.

\section*{Appendix C - Proof of Theorem \ref{Theorem_t_Distribution_Order_1}}
\renewcommand{\theequation}{C.\arabic{equation}}
\setcounter{equation}{0}

Recall from \eqref{ref4} that
\begin{align}
&- \Exp\left[\log q_{\mbox{\tiny t}}(\tilde{\bY}_{n})\right] \nn \\
%%%%%%%%%%%%%%%%%%%%%%%%%%%%%%%%%%%%%%%%%%%%%%%%%
\label{ref4b}
&\leq -\log C_{n} +\tfrac{1}{2}\log |\bSigma|
+ \tfrac{1}{2}(\nu+n) \log \left[1+\tfrac{1}{\nu} \Exp \left\{ (\tilde{\bY}_{n}-\bmu)^{\mbox{\tiny T}}\bSigma^{-1}(\tilde{\bY}_{n}-\bmu)\right\} \right].
\end{align}

In order to confine ourself to a set of only two free parameters, let us consider a tridiagonal Toeplitz matrix of the form 
\begin{align}
\label{ref6b}
\bSigma^{-1} = 
\begin{pmatrix}
\alpha   & \beta    & 0        & \cdots & 0       \\
\beta    & \alpha   & \beta    & \cdots & 0       \\
0        & \beta    &\alpha    & \cdots & 0       \\
\vdots   & \vdots   & \vdots   & \ddots & \vdots  \\
0        & 0        & 0        & \cdots & \alpha 
\end{pmatrix},
\end{align} 
where $\alpha>0$. Since every positive definite matrix is invertible and its inverse is also positive definite \cite[p.\ 438, Theorem 7.2.1]{MATRIX_Analysis}, it is enough to require that $\bSigma^{-1}$ is positive definite.
From the sufficient conditions in Proposition \ref{Toeplitz_Lemma}, we conclude that $\bSigma^{-1}$ is positive definite as long as $\beta \in (-\alpha/2,\alpha/2)$. 
Due to the fact that the components of $\bU_{n}$ are uniformly distributed in $[0,1)$, let us choose $\bmu = [\mu_{\mbox{\tiny Y}}+\tfrac{1}{2}~\mu_{\mbox{\tiny Y}}+\tfrac{1}{2}~\ldots~\mu_{\mbox{\tiny Y}}+\tfrac{1}{2}]^{\mbox{\tiny T}}$, where $\mu_{\mbox{\tiny Y}}$ is the expectation of the process $\{Y_{n}\}$, and then
\begin{align}
&\Exp \left\{ (\tilde{\bY}_{n}-\bmu)^{\mbox{\tiny T}}\bSigma^{-1}(\tilde{\bY}_{n}-\bmu)\right\} \nn \\
\label{ref7b}
&~~~=  
\alpha \sum_{i=1}^{n} \Exp \left\{(\tilde{Y}_{i}-\mu_{\mbox{\tiny Y}}-\tfrac{1}{2})^{2} \right\} 
+ 2 \beta \sum_{i=1}^{n-1} \Exp \left\{ (\tilde{Y}_{i}-\mu_{\mbox{\tiny Y}}-\tfrac{1}{2}) (\tilde{Y}_{i+1}-\mu_{\mbox{\tiny Y}}-\tfrac{1}{2}) \right\}. 
\end{align}
Substituting \eqref{help4} and \eqref{help5} back into \eqref{ref7b} yields
\begin{align}
\Exp \left\{ (\tilde{\bY}_{n}-\bmu)^{\mbox{\tiny T}}\bSigma^{-1}(\tilde{\bY}_{n}-\bmu)\right\}
&= \alpha n \left(R_{\mbox{\tiny Y}}(0) + \tfrac{1}{12}\right) + 2\beta(n-1) R_{\mbox{\tiny Y}}(1) \\
&\leq n \left[\alpha \left(R_{\mbox{\tiny Y}}(0) + \tfrac{1}{12}\right) + 2\beta R_{\mbox{\tiny Y}}(1) \right] \\
\label{help1b}
&= n f(\alpha,\beta).
\end{align}
Since $\bSigma^{-1}$ is positive definite, its determinant is nonzero, and the determinant of $\bSigma$ is given by
\begin{align}
\det(\bSigma) = [\det(\bSigma^{-1})]^{-1}.
\end{align} 
In order to evaluate the determinant of $\bSigma^{-1}$ as a function of the dimension $n$, we denote its determinant by $\phi_{n}$. Note that a recursion relation between the consecutive $\phi_{n}$ is given by
\begin{align}
\label{ref8b}
\phi_{n+2} = \alpha \phi_{n+1} - \beta^{2} \phi_{n},
\end{align}   
with the initial conditions of 
\begin{align}
\phi_{0} = 1;~~~\phi_{1} = \alpha.
\end{align}
The second order difference equation in \eqref{ref8b} can be solved by invoking Z-transform techniques to yield the closed form expression of 
\begin{align}
\phi_{n} = \frac{I^{n+1}-J^{n+1}}{\sqrt{\alpha^{2}-4\beta^{2}}},
\end{align} 
where,
\begin{align}
I = I(\alpha,\beta) = \frac{\alpha + \sqrt{\alpha^{2}-4\beta^{2}}}{2};~~~
J = J(\alpha,\beta) = \frac{\alpha - \sqrt{\alpha^{2}-4\beta^{2}}}{2}.
\end{align}
Thus,
\begin{align}
\det(\bSigma) 
&=\frac{\sqrt{\alpha^{2}-4\beta^{2}}}{I^{n+1}-J^{n+1}}\\
\label{ref9b}
&\leq\frac{\sqrt{\alpha^{2}-4\beta^{2}}}{I^{n+1}-\tfrac{1}{2}I^{n+1}}\\
\label{help2b}
&=\frac{2\sqrt{\alpha^{2}-4\beta^{2}}}{I^{n+1}},
\end{align}
where \eqref{ref9b} holds for all sufficiently large $n$, since $I > J$. We now upper-bound \eqref{ref4b} using \eqref{help1b} and \eqref{help2b}, and also substitute the expression for the normalizing constant from \eqref{Normalizing}. We divide by $n$ and arrive at    
\begin{align}
&-\frac{1}{n} \Exp\left[\log q_{\mbox{\tiny t}}(\tilde{\bY}_{n})\right] \nn \\
%%%%%%%%%%%%%%%%%%%%%%%%%%%%%%%%%%%%%%%%%%%%%%%%%
%%%%%%%%%%%%%%%%%%%%%%%%%%%%%%%%%%%%%%%%%%%%%%%%%
\label{ref10b}
&\leq - \frac{1}{n} \log\Gamma[(\nu+n)/2]
+ \frac{1}{n} \log\Gamma(\nu/2)
+ \frac{1}{2} \log(\nu \pi) \nn \\
&~~~~~~
+\frac{1}{2n}\log \left(2\sqrt{\alpha^{2}-4\beta^{2}}\right)
-\frac{n+1}{2n}\log I(\alpha,\beta)
+ \frac{1}{2}\frac{\nu+n}{n} \log \left(1+\frac{n}{\nu} f(\alpha,\beta) \right).
\end{align}
Upper-bounding \eqref{ref10b} using \eqref{help6} and some basic algebra yields that
\begin{align}
&-\frac{1}{n} \Exp\left[\log q_{\mbox{\tiny t}}(\tilde{\bY}_{n})\right] \nn \\
%%%%%%%%%%%%%%%%%%%%%%%%%%%%%%%%%%%%%%%%%%%%%%%%%
%%%%%%%%%%%%%%%%%%%%%%%%%%%%%%%%%%%%%%%%%%%%%%%%%
&\leq \frac{1}{2n} \log\left(\frac{\nu+n}{2}\right) + \frac{\nu+n}{2n} - \frac{1}{2n} \log (2\pi)
+ \frac{1}{n} \log\Gamma(\nu/2)
+ \frac{1}{2} \log(\nu \pi) \nn \\
&~~~~~~
+\frac{1}{2n}\log \left(2\sqrt{\alpha^{2}-4\beta^{2}}\right)
-\frac{n+1}{2n}\log I(\alpha,\beta)
+ \frac{1}{2}\frac{\nu+n}{n} \log \left(\frac{2+\frac{2n}{\nu} f(\alpha,\beta)}{\nu+n} \right).
\end{align}
Taking now the limit, we find that the entropy rate is upper-bounded by
\begin{align}
\bar{H}(Y) 
&\leq  
\frac{1}{2} 
+ \frac{1}{2} \log(\nu \pi) 
-\frac{1}{2}\log I(\alpha,\beta)
+ \frac{1}{2}\log \left(\frac{2}{\nu} f(\alpha,\beta) \right) \\
&= \frac{1}{2} \log \left[2\pi e \frac{f(\alpha,\beta)}{I(\alpha,\beta)}\right].
\end{align}
As a final step, we may optimize over $\alpha$ and $\beta$ to get the tightest possible bound:
\begin{align} \label{ToRef3}
\bar{H}(Y) 
\leq \inf_{\alpha > 0} \inf_{\beta \in (-\alpha/2,\alpha/2)} \frac{1}{2} \log \left(2\pi e \frac{f(\alpha,\beta)}{I(\alpha,\beta)}\right).
\end{align}
Following similar considerations as at the end of the proof of Theorem \ref{Theorem_t_Distribution_Order_K}, we are able to simplify the minimization problem in \eqref{ToRef3} and arrive at    
\begin{align} 
\bar{H}(Y) 
\leq \inf_{s \in (-1,1)} \frac{1}{2} \log \left(4\pi e \frac{\left(R_{\mbox{\tiny Y}}(0) + \tfrac{1}{12}\right) + s R_{\mbox{\tiny Y}}(1)}{1 + \sqrt{1-s^{2}}}\right).
\end{align}

\section*{Appendix D - Proof of Lemma \ref{Lemma_Poisson}}
\renewcommand{\theequation}{D.\arabic{equation}}
\setcounter{equation}{0}

By Poisson's integral formula we have 
\begin{align}
\frac{1}{2\pi} \int_{0}^{2\pi} \log \left|e^{i\lambda}+\alpha\right|^{2} \dint\lambda = 0,
\end{align}
if $|\alpha| \leq 1$. This is equivalent to 
\begin{align}
\frac{1}{2\pi} \int_{0}^{2\pi} \log \left(1+\alpha^{2} + 2\alpha \cos(\lambda)\right) \dint\lambda = 0,
\end{align}
or, to, 
\begin{align}
\log \left(1+\alpha^{2}\right) + 
\frac{1}{2\pi} \int_{0}^{2\pi} \log \left(1 + \frac{2\alpha}{1+\alpha^{2}} \cos(\lambda)\right) \dint\lambda = 0.
\end{align}
For $-1 \leq \alpha \leq 1$, the function $t(\alpha) = \tfrac{2\alpha}{1+\alpha^{2}}$ is monotonically increasing in the range $[-1,1]$ and its inverse is given by $\alpha(t) = \tfrac{1-\sqrt{1-t^{2}}}{t}$. Elementary algebra yields that 
\begin{align}
1+\alpha^{2}(t) = \frac{2-2\sqrt{1-t^{2}}}{t^{2}},
\end{align}
and thus, for any $t \in [-1,0)\cup(0,1]$,
\begin{align}
\frac{1}{2\pi} \int_{0}^{2\pi} \log \left(1 + t \cos(\lambda)\right) \dint\lambda = -\log \left(\frac{2-2\sqrt{1-t^{2}}}{t^{2}}\right),
\end{align}
and specifically for $t=0$,
\begin{align}
\frac{1}{2\pi} \int_{0}^{2\pi} \log \left(1 + t \cos(\lambda)\right) \dint\lambda = 0.
\end{align}

\end{document}